\documentclass[11pt]{article}
\usepackage{datetime}
\usepackage{color,array,graphics}
\usepackage{enumerate}
\usepackage[papersize={8.5in,11in},top=1in,left=1in,right=1in,bottom=1in]{geometry}

\usepackage{amsmath}
\usepackage{amsthm}
\usepackage{amssymb}
\usepackage{mathrsfs}
\usepackage{multirow}
\usepackage{tikz}
\usepackage{sectsty}
\usepackage{float}
\usepackage{cite}
\usepackage{bm}
\usepackage{hyperref}

\newtheorem{theorem}{Theorem}

\newtheorem{lemma}{Lemma}
\newtheorem{definition}{Definition}
\newtheorem{proposition}{Proposition}

\usepackage[noend]{algpseudocode}
\usepackage[ruled]{algorithm2e}

\date{}

\begin{document}

\title{Tradeoffs Between Information and Ordinal Approximation for Bipartite Matching\footnote{This work was partially supported by NSF award CCF-1527497.} }

%



\makeatletter
\renewcommand\@date{{%
  \vspace{-\baselineskip}%
  \large\centering
  \begin{tabular}{@{}c@{}}
    Elliot Anshelevich \\
    \normalsize eanshel@cs.rpi.edu
  \end{tabular}%
  \quad \quad
  \begin{tabular}{@{}c@{}}
    Wennan Zhu \\
    \normalsize zhuw5@rpi.edu
  \end{tabular}

  \bigskip
  
  Rensselaer Polytechnic Institute\\
  110 8th Street, Troy NY 12180\\

}}
\makeatother

\maketitle

\begin{abstract} We study ordinal approximation algorithms for maximum-weight bipartite matchings. Such algorithms only know the ordinal preferences of the agents/nodes in the graph for their preferred matches, but must compete with fully omniscient algorithms which know the true numerical edge weights (utilities). 
Ordinal approximation is all about being able to produce good results with only limited information. Because of this, one important question is how much better the algorithms can be as the amount of information increases. To address this question for forming high-utility matchings between agents in $\mathcal{X}$ and $\mathcal{Y}$, we consider three ordinal information types: when we know the preference order of only nodes in $\mathcal{X}$ for nodes in $\mathcal{Y}$,
when we know the preferences of both $\mathcal{X}$ and $\mathcal{Y}$,
and when we know the total order of the edge weights in the entire graph, although not the weights themselves.  We also consider settings where only the top preferences of the agents are known to us, instead of their full preference orderings. We design new ordinal approximation algorithms for each of these settings, and quantify how well such algorithms perform as the amount of information given to them increases.
\end{abstract}

\section{Introduction}
\label{section-intro}
Many important settings involve agents with preferences for different outcomes. Such settings include, for example, social choice and matching problems. Although the quality of an outcome to an agent may be measured by a numerical utility, it is often not possible to obtain these exact utilities when forming a solution. This can occur because eliciting numerical information from the agents may be too difficult, the agents may not want to reveal this information, or even because the agents themselves do not know the exact numerical values. On the other hand, eliciting {\em ordinal} information (i.e., the preference ordering of each agent over the outcomes) is often much more reasonable. Because of this, there has been a lot of recent work on {\em ordinal approximation algorithms}: these are algorithms which only use ordinal preference information as their input, and yet return a solution provably close to the optimum one (e.g., \cite{filos2014social,christodoulou2016social,anshelevich2015blind,anshelevich2016truthful,anshelevich2015approximating,skowron2017social,goel2016metric,feldman2016}). In other words, these are algorithms which only use limited ordinal information, and yet can compete in the quality of solution produced with omniscient algorithms which know the true (possibly latent) numerical utility information.


Ordinal approximation is all about being able to produce good results with only limited information. Because of this, it is important to quantify how well algorithms can perform as more information is given. If the quality of solutions returned by ordinal algorithms greatly improves when they are provided more information, then it may be worthwhile to spend a lot of resources in order to acquire such more detailed information. If, on the other hand, the improvement is small, then such an acquisition of more detailed information would not be worth it. Thus the main question we consider in this paper is: {\em How does the quality of ordinal algorithms improve as the amount of information provided increases?}

In this paper, we specifically consider this question in the context of computing a maximum-utility matching in a metric space. Matching problems, in which agents have preferences for which other agents they want to be matched with, are ubiquitous. The maximum-weight metric matching problem specifically provides solutions to important applications, such as forming diverse teams and matching in friendship networks (see \cite{anshelevich2015blind,anshelevich2016truthful} for much more discussion of this). Formally, there exists a complete undirected bipartite graph for two sets of agents $\mathcal{X}$ and $\mathcal{Y}$ of size $N$, with an edge weight $w(x,y)$ representing how much utility $x\in\mathcal{X}$ and $y\in\mathcal{Y}$ derive from their match; these edge weights satisfy the triangle inequality. The algorithms we consider, however, do not have access to such numerical edge weights: they are only given ordinal information about the agent preferences. The goal is to form a perfect matching between $\mathcal{X}$ and $\mathcal{Y}$, in order to approximate the maximum weight matching as much as possible using only the given ordinal information. We compare the weight of the matching returned by our algorithms with the true maximum-weight perfect matching in order to quantify the performance of our ordinal algorithms. 

\vskip 5pt\noindent{\bf Types of Ordinal Information~~}
Ordinal approximation algorithms for maximum weight matching have been considered before in \cite{anshelevich2015blind,anshelevich2016truthful}, although only for complete graphs; algorithms for bipartite graphs require somewhat different techniques. Our main contribution, however, lies in considering many types of ordinal information, forming different algorithms for each, and quantifying how much better types of ordinal information improve the quality of the matching formed. Specifically, we consider the following types of ordinal information.


\begin{itemize}

\item The most restrictive model we consider is {\em one-sided preferences}. That is, only preferences for agents in $\mathcal{X}$ over agents in $\mathcal{Y}$ are given to our algorithm. These preferences are assumed to be consistent with the (hidden) agent utilities, i.e., if $x$ prefers $y_1$ to $y_2$, then it must be that $w(x,y_1)\geq w(x,y_2)$. Such one-sided preferences may occur, for example, when $\mathcal{X}$ represents people and $\mathcal{Y}$ represents houses. People have preferences over different houses, but houses do not have preferences over people. These types of preferences also apply to settings in which both sides have preferences, but we only have access to the preferences of $\mathcal{X}$, e.g., because the agents in $\mathcal{Y}$ are more secretive. 

\item The next level of ordinal information we consider is {\em two-sided preferences}, that is, both preferences for agents in $\mathcal{X}$ over $\mathcal{Y}$ and agents in $\mathcal{Y}$ over $\mathcal{X}$ are given. This setting could apply to the situation that two sets of people are collaborating, and they have preferences over each other, or of a matching between job applicants and possible employers. As we consider the model in a metric space, the distance (weight) between two people could represent the diversity of their skills, and a person prefers someone with most diverse skills from him/her in order to achieve the best results of collaboration. 

\item The most informative model which we consider in this paper is that of {\em total-order}. That is, the order of all the edges in the bipartite graph is given to us, instead of only local preferences for each agent. In this model, global ordinal information is available, compared to the preferences of each agent in the previous two models. Studying this setting quantifies how much efficiency is lost due to the fact that we only know ordinal information, as opposed to the fact that we only know {\em local} information given to us by each agent. 
\end{itemize}

Comparing the results for the above three information types allows us to answer questions like: ``Is it worth trying to obtain two-sided preference information or total order information when only given one-sided preferences?" However, above we always assumed that for an agent $x$, we are given their entire preferences for all the agents in $\mathcal{Y}$. Often, however, an agent would not give their preference ordering for all the agents they could match with, and instead would only give an ordered list of their top preferences. Because of this, in addition to the three models described above, we also consider the case of {\em partial} ordinal preferences, in which only the top $\alpha$ fraction of a preference list is given by each agent of $\mathcal{X}$. Thus for $\alpha=0$ no information at all is given to us, and for $\alpha=1$ the full preference ordering of an agent is given.
Considering partial preferences tells us when, if there is a cost to buying information, we might choose to buy only part of the ordinal preferences. We establish tradeoffs between the percentage of available preferences and the possible approximation ratio for all three models of information above, and thus quantify when a specific amount of ordinal information is enough to form a high-quality matching.

\vskip 5pt\noindent{\bf Our Contributions~~}
We show that as we obtain more ordinal information about the agent preferences, we are able to form better approximations to the maximum-utility matching, even without knowing the true numerical edge weights. 
Our main results are shown in Figure \ref{fig:plot_figures}.


\begin{figure}[H]
  \centerline{
  \begin{minipage}[c]{0.55\textwidth}
    \includegraphics[width=\textwidth]{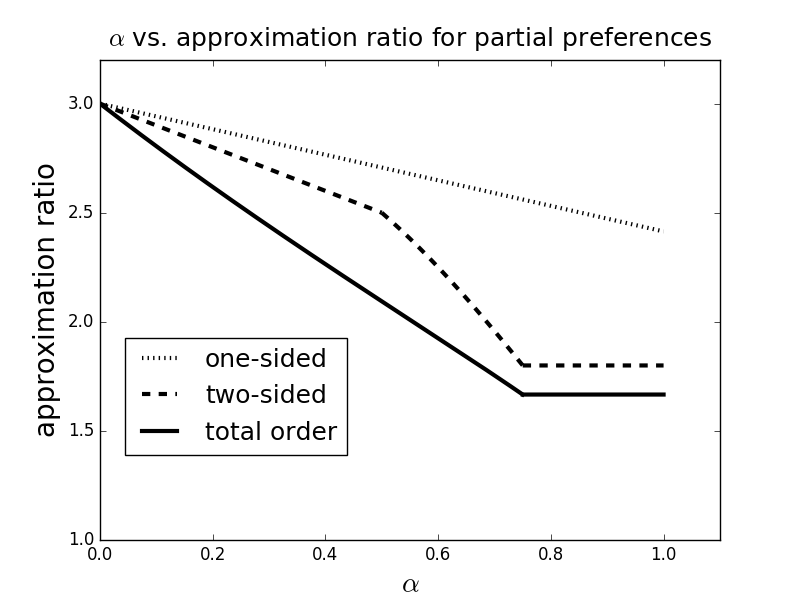}
  \end{minipage}\hfill
  \begin{minipage}[c]{0.45\textwidth}
    \caption{$\alpha$ vs. approximation ratio for partial information. As we obtain more information about the agent preferences ($\alpha$ increases), we are able to form better approximation to the maximum-weight matching. The tradeoff for one-sided preferences is linear, while it is more complex for two-sided and total order.} \label{fig:plot_figures}
  \end{minipage}}
\end{figure}

Using only one-sided preference information, with only the order of top $\alpha N$ preferences given for agents in $\mathcal{X}$, we are able to form a $(3 - (2 - \sqrt{2})\alpha)$-approximation. We do this by combining random serial dictatorship with purely random matchings. When $\alpha = 1$, the algorithm yields a $(\sqrt{2}+1)$-approximation. This is the first non-trivial analysis for the performance of $RSD$ on maximum bipartite matching in a metric space, and this analysis is one of our main contributions.

Given two-sided information, with the order of top $\alpha N$ preferences for agents in both $\mathcal{X}$ and $\mathcal{Y}$, we can do significantly better. When $\alpha \ge \frac{1}{2}$, adopting an existing framework in \cite{anshelevich2015blind}, by mixing greedy and random algorithms, and adjusting it for bipartite graphs, we get a $\frac{(3 - 2\alpha)(3 - \alpha)}{2\alpha^2-3\alpha+3}$-approximation. When $\alpha \le \frac{1}{2}$, the framework would still work, but would not produce a good approximation. We instead design a different algorithm to get better results. Inspired by $RSD$, we take advantage of the information of preferences from both sets of agents, adjust $RSD$ to obtain ``undominated" edges in each step, and finally combine it with random matchings to get a $(3 - \alpha)$-approximation. When $\alpha \ge \frac{3}{4}$, the algorithm yields a $1.8$-approximation.

For the total-ordering model, the order of top $\alpha N^2$ heaviest edges in the bipartite graph is given. We use the framework in \cite{anshelevich2015blind} again to obtain a $\frac{2 + \sqrt{1 - \alpha}}{2 - \sqrt{1 - \alpha}}$-approximation. Here we must re-design the framework to deal with the cases that $\alpha \le \frac{3}{4}N$, which is not a straight-forward adjustment. When $\alpha \ge \frac{3}{4}N$ the algorithm yields a $\frac{5}{3}$-approximation.



Finally, in Section \ref{section-onesided-restricted} we analyze the case when edge weights cannot be too different: the highest weight edge is at most $\beta$ times the lowest weight edge in one-sided model. When the edge weights have this relationship, we can extend our analysis to give a $(\sqrt{\beta - \frac{3}{4}} + \frac{1}{2})$-approximation, even without assuming that edge weights form a metric.


\vskip 5pt\noindent{\bf Discussion and Related Work~~}
Previous work on forming good matchings can largely be classified into the following classes. First, there is a large body of work assuming that numerical weights or utilities don't exist, only ordinal preferences. Such work studies many possible objectives, such as forming stable matchings (see e.g., \cite{roth1992two,rastegari2013two}), or maximizing objectives determined only by the ordinal preferences (e.g., \cite{abrahamIKM07,chakrabartyS14}). Second, there is work assuming that numerical utilities or weights exist, and are {\em known} to the matching designer. 
Unlike the above two settings, we consider the case when numerical weights {\em exist}, but are latent or {\em unknown}, and yet the goal is to approximate the true social welfare, i.e., maximum weight of a perfect matching. Note that although some previous work assumes that all numerical utilities are known, they often still use algorithms which only require ordinal information, and thus fit into our framework; we discuss some of these results below.




Similar to our one-sided model, house allocation \cite{abdulkadirouglu1998random} is a popular model of assigning $n$ agents to $n$ items. 
\cite{bhalgat2011social} studied the ordinal welfare factor and the linear welfare factor of RSD and other ordinal algorithms. \cite{krysta2014size} studied both maximum matching and maximum vertex weight matching using an extended RSD algorithm. These either used objectives depending only on ordinal preferences, such as the size of the matching formed, or used node weights (as opposed to edge weights). \cite{filos2014social} and \cite{christodoulou2016social} assumed the presence of numerical agent utilities and studied the properties of RSD. Crucially, this work assumed normalized agent utilities, such as unit-sum or unit-range. This allowed \cite{filos2014social,christodoulou2016social} to prove approximation ratios of $\Theta(\sqrt{n})$ for RSD. Instead of assuming that agent utilities are normalized, we consider agents in a metric space; this different correlation between agent utilities allows us to prove much stronger results, including a constant approximation ratio for RSD. Kalyanasundaram et al. studied serial dictatorship for maximum weight matching in a metric space \cite{kalyanasundaram1991line}, and gave a 3-approximation for RSD in this, while we are able to get a tighter bound of 2.41-approximation.\footnote{Note that many of the papers mentioned here specifically attempt to form {\em truthful} algorithms. While RSD is certainly truthful, in this paper we attempt to quantify what can be done using ordinal information in the presence of latent numerical utilities, and leave questions of truthfulness to future work.}


Besides maximizing social welfare, minimizing the social cost of a matching is also popular. \cite{caragiannis2016truthful} studied the approximation ratio of RSD and augmentation of serial dictatorship (SD) for minimum weight matching in a metric space. Their setting is very similar to ours, except that we consider the maximization problem, which has different applications \cite{anshelevich2015blind,anshelevich2016truthful}, and allows for a much better approximation factor (constant instead of linear in $n$) using different techniques.


Another area studying ordinal approximation algorithms is social choice, where the goal is to decide a single winner in order to maximize the total social welfare. This is especially related to our work when the hidden utilities of voters are in a metric space (see e.g., \cite{anshelevich2015approximating,skowron2017social,goel2016metric,feldman2016}),

The work most related to ours is \cite{anshelevich2015blind,anshelevich2016truthful}. As mentioned above, we use an existing framework \cite{anshelevich2015blind} for the two-sided and the total-order model. While the goal is the same: to approximate the maximum weight matching using ordinal information, this paper is different from \cite{anshelevich2015blind} in several aspects. \cite{anshelevich2015blind} only considered approximating the true maximum weight matching for non-bipartite complete graphs. We instead focus on bipartite graphs, and especially on considering different levels of ordinal information by analyzing three models with increasing amount of information, and also consider partial preferences. Although we use similar techniques for parts of two-sided and total-order model analysis, they need significant adjustments to deal with bipartite graphs and partial preferences; moreover, the method used for analyzing the one-sided model is quite different from \cite{anshelevich2015blind}. 

\section{Model and Notation}
\label{section-model}
For all the problems studied in this paper, we are given as input two sets of agents  $\mathcal{X}$ and $\mathcal{Y}$ with $|\mathcal{X}| = |\mathcal{Y}| = N $. $G = (\mathcal{X}, \mathcal{Y}, E)$ is an undirected complete bipartite graph with weights on the edges. 
We assume that the agent preferences are derived from a set of underlying hidden edge weights $w(x, y)$ for each edge $(x, y)$, $x \in \mathcal{X}, y \in \mathcal{Y}$. $w(x,y)$ represents the utility of the match between $x$ and $y$, so if $x$ prefers $y_1$ to $y_2$, then it must be that $w(x,y_1)\geq w(x,y_2)$. Let $OPT(G)$ denote the complete bipartite matching that gives the maximum total edge weights. $w(G)$ of any bipartite graph $G$ is the total edge weight of the graph, and $w(M)$ of any matching $M$ is the total weight of edges in the matching. The agents lie in a metric space, by which we will only mean that, $\forall x_1, x_2 \in \mathcal{X}, \forall y_1, y_2 \in \mathcal{Y}, w(x_1, y_1) \le w(x_1, y_2) + w(x_2, y_1) + w(x_2, y_2)$. We assume this property in all sections except for Section~\ref{section-onesided-restricted}.

For the setting of one-sided preferences, $\forall x \in \mathcal{X}$, we are given a strict preference ordering $P_x$ over the agents in $\mathcal{Y}$. When dealing with partial preferences, only top $\alpha N$ agents in $P_x$ are given to us in order. We assume $\alpha N$ is an integer, $\alpha \in [0, 1]$. Of course, when $\alpha=0$, nothing can be done except to form a completely random matching. For two-sided partial preferences, we are given both the top $\alpha$ fraction of preferences $P_x$ of agents $x$ in $\mathcal{X}$ over those in $\mathcal{Y}$, and vice versa. For the total order setting, we are given the order of the highest-weight $\alpha N^2$ edges in the complete bipartite graph $G=(\mathcal{X}, \mathcal{Y}, E)$.

\section{One-sided Ordinal Preferences}\label{section-onesided}


For one-sided preferences, our problem becomes essentially a house allocation problem to maximize social welfare, see e.g., \cite{christodoulou2016social,krysta2014size,filos2014social}. 
Before we proceed, it is useful to establish a baseline for what approximation factor is reasonable. Simply picking a matching uniformly at random immediately results in a 3-approximation (see Theorem \ref{perfect-random-bound}), and there are examples showing that this bound is tight. Other well-known algorithms, such as Top Trading Cycle, also cannot produce better than a 3-approximation to the maximum weight matching for our setting. Serial Dictatorship, which uses only one-sided ordinal information, is also known to give a 3-approximation to the maximum weight matching for our problem~\cite{kalyanasundaram1991line}. Serial Dictatorship simply takes an arbitrary agent from $x\in\mathcal{X}$, assigns it $x$'s favorite unallocated agent from $\mathcal{Y}$, and repeats. Unfortunately, it is not difficult to show that this bound of 3 is tight. Our first major result in this paper is to prove that {\em Random} Serial Dictatorship always gives a $(\sqrt{2} + 1)$-approximation in expectation, no matter what the true numerical weights are, thus giving a significant improvement to all the algorithms mentioned above.

\vspace{3mm}
\begin{algorithm}[H]
\caption{Random Serial Dictatorship for Perfect Matching of one-sided ordering.}
\label{onesided-rsd}
  Initialize $M =  \emptyset $, $G = (\mathcal{X}, \mathcal{Y}, E)$ \;
   \While {$E \ne \emptyset$} {
     Pick an agent $x$ uniformly at random from $\mathcal{X} $ \;
     Let $y$ denote $x$'s most preferred agent in $\mathcal{Y} $ \;
     Take  $e = (x,y)$ from $E$ and add it to $M$ \;
     Remove $x$, $y$, and all edges containing $x$ or $y$ from the graph $G$ \;
    }
  \textbf{Final Output:} Return $M$.
\end{algorithm}
\vspace{3mm}


\begin{theorem}
\label{onesided-rsd-bound}
  Suppose $G=(\mathcal{X}, \mathcal{Y}, E)$ is a complete bipartite graph on the set of nodes $\mathcal{X}, \mathcal{Y}$ with $|\mathcal{X}| = |\mathcal{Y}| = N$. Then, the expected weight of the perfect matching $M$ returned by Algorithm~\ref{onesided-rsd} is $E[w(M)] \ge \frac{1}{\sqrt{2}+1}w(OPT(G))$.
\end{theorem}

\begin{proof}
  \textbf{Notation}: Consider a bipartite subgraph $S \subseteq G$, that satisfies $S = (\mathcal{X}', \mathcal{Y}', E'), \ \mathcal{X}' \subseteq \mathcal{X}, \   \mathcal{Y}' \subseteq \mathcal{Y}$, and $|\mathcal{X}'| = |\mathcal{Y}'|$. Let $Min(S)$ denote a {\em minimum} weight perfect matching on $S$, and $RSD(S)$ denote the expected weight returned by Algorithm~\ref{onesided-rsd} on graph $S$.

  For any $x \in \mathcal{X}'$, we use $\lambda(S,x)$ to denote the edge between $x$ and its most preferred agent in $\mathcal{Y}'$. Define $R(S, x)$ as the remaining graph after removing $x$, $x$'s most preferred agent, and all the edges containing $x$ or $x$'s most preferred agent from $S$.

  We begin by simply expressing $RSD(S)$ in terms of these quantities.
  \begin{lemma}
  \label{onesided-weight}
    For any subgraph $S$ as decribed above, \\
    $RSD(S) = \frac{1}{|\mathcal{X}'|} \sum_{x \in \mathcal{X}'} w(\lambda(S,x)) + \frac{1}{|\mathcal{X}'|} \sum_{x \in \mathcal{X}'} RSD(R(S, x))$.
  \end{lemma}
  \begin{proof}
    This simply follows from definition of expectation. In the first round of Algorithm~\ref{onesided-rsd}, an agent $x$ is selected uniformly at random from $\mathcal{X}'$. Given that $x$ is selected, the edge added to the matching is exactly $\lambda(S,x)$, and the expected weight of the matching for the remaining graph is exactly $RSD(R(S,x))$. Each of these occurs with probability $1/|\mathcal{X}'|$.
  \end{proof}

  We now state the main technical lemma which allows us to prove the result. This lemma gives a bound on the maximum weight matching in terms of the quantities defined above.

  \begin{lemma}
    \label{onesided-structural}
    For any given graph $G = (\mathcal{X}, \mathcal{Y}, E)$, one of the following two cases must be true:\\
    \textbf{Case 1}:  $w(OPT(G)) \le \frac{1}{|\mathcal{X}|} \sum_{x \in \mathcal{X}} w(OPT(R(x))) + \frac{\sqrt{2}+1}{|\mathcal{X}|}\sum_{x \in \mathcal{X}} w(\lambda(x))$\\
    \textbf{Case 2}: $w(OPT(G)) \le (\sqrt{2}+1)w(Min(G))$
  \end{lemma}

  We will prove this lemma below, but first we discuss how the rest of the proof will proceed. When Case 1 above holds, we know that at any step of the algorithm, the change in the weight of the optimum solution in the remaining graph is not that different from the weight of the edge selected by our algorithm. This allows us to compare the weight of $OPT$ with the weight of the matching returned by our algorithm. In fact, this is the technique used in a previous paper~\cite{anshelevich2016truthful} to analyze RSD for complete graphs (i.e., non-bipartite graphs), and show that RSD gives a 2-approximation for perfect matching on complete graphs.
  Similar to Case 1 in Lemma~\ref{onesided-structural}, this was done by proving that in each step, the expected loss of optimal matching is at most twice the expected weight of the chosen edge, and thus the emtire algorithm gives a 2-approximation.

  It is important to note here that this {\em does not} work for bipartite graphs. In bipartite matching, using only this method will not give an approximation ratio better than 3. To see this, consider the bipartite graph in Figure~\ref{fig:rsd-example}. Suppose $G = (\mathcal{X}, \mathcal{Y}, E)$ is a complete bipartite graph, $|\mathcal{X}| = |\mathcal{Y}| = N$. The edges shown in the Figure are the maximum weight matching of $G$; all the other edges have weight of 1. It is easy to see that these edge weights form a metric. $\forall x \in \mathcal{X}$, $x$'s most preferred agent in $\mathcal{Y}$ is $y_1$, second preferred agent is $y_2$, ..., least preferred agent is $y_n$ (we can always perturb the edge weights by an infinitesimal amount to remove ties for this example). Then the weight of the optimum solution is $w(OPT(G)) = (N-1) + 3$. In this example, the expected decrease in the weight of the optimal matching in the first step of RSD is $3$: choosing $x_1$ loses 3, and choosing any other agent $x_i$ in $\mathcal{X}$ loses 3 since $(x_1,y_1)$ and $(x_i,y_i)$ can no longer ne used (decrease of 4), but the edge $(x_1,y_i)$ can be used (increase of 1). On the other hand, the expected weight of the edge chosen by RSD is $\frac{3+(N-1)}{N}$. In this case, almost 3 times the expected weight of the chosen edge is needed to compensate for the loss of optimal matching, so the inequality in ``Case 1" above only holds if we replace $\sqrt{2}+1$ with 3, and thus would only result in a 3-approximation.

   \begin{figure}[H]
 \begin{center}
 \includegraphics[scale=0.5]{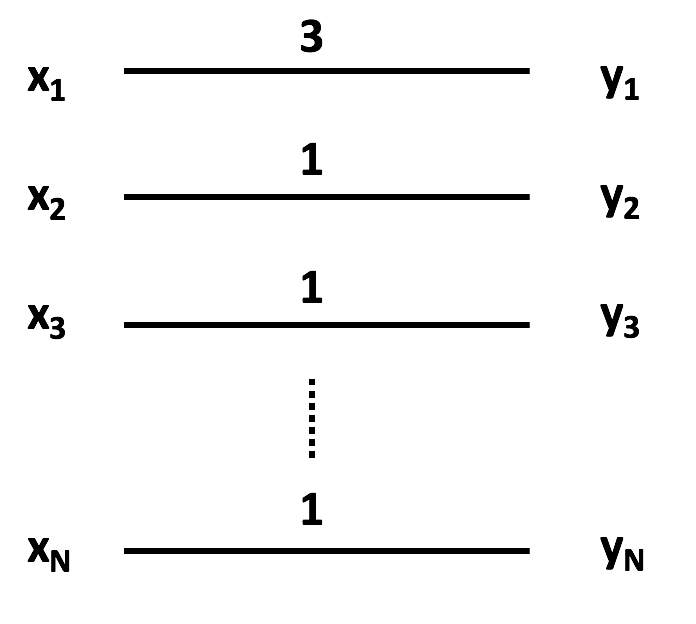}
 \end{center}
 \caption{An example graph for RSD.}
 \label{fig:rsd-example}
 \end{figure}

 We get around this problem by adding Case 2 to our lemma. We are able to show that in any step, either the expected loss of the weight of the optimal matching is at most $(\sqrt{2}+1)$ times of expected weight of the chosen edge, or the weight of the optimal matching is at most $(\sqrt{2}+1)$ times weight of the minimum weight matching. By combining these two cases, we can easily show the following claim which completes the proof of the theorem.

   \begin{proposition}
   \label{onesided-proposition}
     As long as Lemma~\ref{onesided-structural} is obeyed for every $S$, Algorithm~\ref{onesided-rsd} provides a $(\sqrt{2}+1)$-approximation to the Maximum weight perfect matching: $RSD(G) \ge \frac{1}{\sqrt{2}+1}w(OPT(G))$.
   \end{proposition}

   \begin{proof}
     We proceed by induction. Clearly when $G$ only has two agents, RSD produces the optimum matching.
     Now consider a bipartite graph $G = (\mathcal{X}, \mathcal{Y}, E)$ with $|\mathcal{X}|=|\mathcal{Y}|=N$, and suppose that the claim is true for all smaller graphs, i.e., $\forall x \in \mathcal{X}$, we know that $RSD(R(G,x)) \ge \frac{1}{\sqrt{2}+1}w(OPT(R(G,x)))$.

     If Case 2 in Lemma~\ref{onesided-structural} holds for $G$, then because $Min(G)$ is the minimum weight perfect matching, we know that $w(Min(G)) \le RSD(G)$. So $RSD(G) \ge \frac{1}{\sqrt{2}+1}w(OPT(G))$. Otherwise Case 1 in Lemma~\ref{onesided-structural} must be true.
     \begin{equation*}
      w(OPT(G)) \le \frac{1}{N} \sum_{x \in \mathcal{X}} w(OPT(R(G,x))) + \frac{\sqrt{2}+1}{N}\sum_{x \in \mathcal{X}} w(\lambda(G,x))
     \end{equation*}

     By our assumption,
     \begin{equation*}
      w(OPT(G)) \le \frac{\sqrt{2}+1}{N} \sum_{x \in \mathcal{X}} RSD(R(G,x)) + \frac{\sqrt{2}+1}{N}\sum_{x \in \mathcal{X}} w(\lambda(G,x))
     \end{equation*}

     This completes the proof by Lemma~\ref{onesided-weight}.
   \end{proof}

We now proceed with the main technical part of the proof, i.e., the proof of Lemma \ref{onesided-structural}.

 \paragraph{Proof of Lemma~\ref{onesided-structural}} For compactness of notation, since $S$ is fixed, we will omit $S$ and simply write $\lambda(x)$ and $R(x)$ instead of $\lambda(S,x)$ and $R(S,x)$. For any fixed $x \in \mathcal{X}'$, denote $x$'s most preferred agent in $\mathcal{Y}'$ as $y$ (so $\lambda(x)=(x,y)$). In $OPT(S)$, suppose $x$ is matched to $b \in \mathcal{Y}'$, and $y$ is matched to $a \in \mathcal{X}'$. In $Min(S)$, suppose $b$ is matched to $m \in \mathcal{X}'$. $\forall x \in \mathcal{X}'$, there exist $y$, $a$, $b$, $m$ as described above. As shown in Figure~\ref{fig:rsd}, denote edge $(x, y)$ by $\lambda(x)$, $(x, b)$ by $P(x)$, $(a, y)$ by $\bar{P}(x)$, and $(a, b)$ by $D(x)$.

  \begin{figure}[htb]
  \begin{center}
  \includegraphics[scale=0.5]{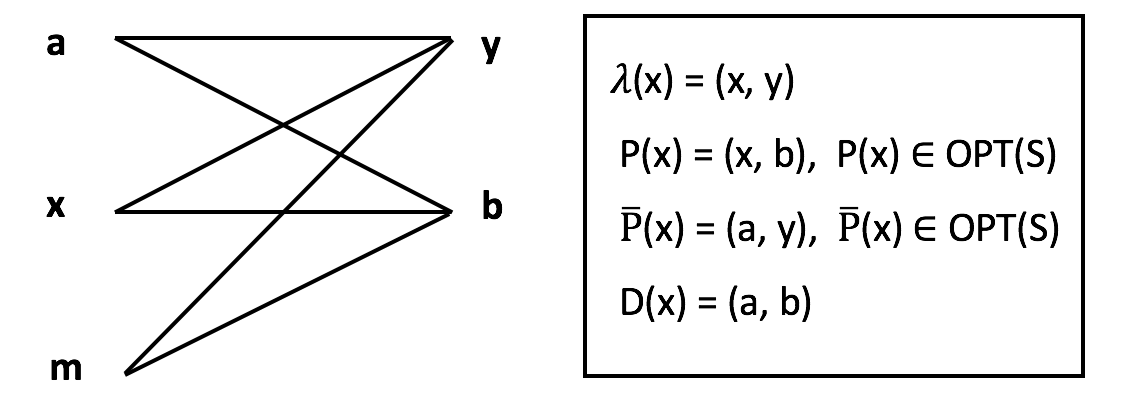}
  \end{center}
  \caption{Notation of $\lambda(x)$, $P(x)$, $\bar{P}(x)$, $D(x)$.}
  \label{fig:rsd}
  \end{figure}

  We'll prove Lemma~\ref{onesided-structural} by showing that if \textbf{Case 2} is not true, then  \textbf{Case 1} must be true. Suppose  \textbf{Case 2} is not true, i.e., $w(OPT(S)) > (\sqrt{2}+1)w(Min(S))$. \\

  Suppose that random serial dictatorship picks $x \in \mathcal{X}'$. Then $OPT(R(S, x))$ is at least as good as the matching obtained by removing $P(x)$ and $\bar{P}(x)$, and adding $D(x)$ to $OPT(S)$ (the rest stay the same):
  \begin{equation*}
  w(OPT(R(x))) \ge w(OPT(S)) - w(P(x)) - w(\bar{P}(x)) + w(D(x))
  \end{equation*}
  Note that when $\lambda(x) \in OPT(S)$, $\bar{P}(x) = P(x) = D(x)$, and the inequality still holds. Summing this up over all nodes $x$, we obtain:

  \begin{align}
  \frac{1}{|\mathcal{X}'|} \sum_{x \in \mathcal{X}'} w(OPT(R(x)))
  & \ge \frac{1}{|\mathcal{X}'|} \sum_{x \in \mathcal{X}'} (w(OPT(S)) - w(P(x)) - w(\bar{P}(x)) + w(D(x))) \nonumber \\
  & = w(OPT(S)) - \frac{1}{|\mathcal{X}'|} \sum_{x \in \mathcal{X}'} (w(P(x)) + w(\bar{P}(x)) - w(D(x))) \nonumber \\
  & = (1-\frac{1}{|\mathcal{X}'|})w(OPT(S)) - \frac{1}{|\mathcal{X}'|} \sum_{x \in \mathcal{X}'} ( w(\bar{P}(x)) - w(D(x))) \label{rsd-eq1}
  \end{align}

  In Figure~\ref{fig:rsd}, by the triangle inequality, we know that
  \begin{equation*}
    w(a, y) \le w(a, b) + w(m, b) + w(m, y)
  \end{equation*}
  Note that when $y = b$ the inequality still holds, because $w(a, y) = w(a, b)$. It also holds when $a = m$ for the same reason.

  Because $\lambda(m)$ is the edge to $m$'s most preferred agent, $w(m, y) \le w(\lambda(m))$, and thus
  \begin{equation*}
    w(\bar{P}(x)) \le w(D(x)) + w(m, b) + w(\lambda(m)))
  \end{equation*}

  Summing this up for all $x \in \mathcal{X}'$, note that each $x$ is matched to a unique $b$ in $OPT(S)$, and each $b$ is matched to a unique $m$ in $Min(S)$, so each agent in $\mathcal{Y}'$ appears as $b$ exactly once and each agent in $\mathcal{X}'$ appears as $m$ exactly once.

  \begin{equation*}
  \sum_{x \in \mathcal{X}'} w(\bar{P}(x)) \le \sum_{x \in \mathcal{X}'} w(D(x)) + w(Min(S)) + \sum_{x \in \mathcal{X}'} w(\lambda(x)))
  \end{equation*}
  \begin{equation}
  \label{rsd-eq2}
  \sum_{x \in \mathcal{X}'} (w(\bar{P}(x)) -  w(D(x)))\le  w(Min(S)) + \sum_{x \in \mathcal{X}'} w(\lambda(x)))
  \end{equation}

  Combining Inequality~\ref{rsd-eq1} and Inequality~\ref{rsd-eq2},
  \begin{equation}
  \label{rsd-eq3}
  \frac{1}{|\mathcal{X}'|} \sum_{x \in \mathcal{X}'} w(OPT(R(x))) \ge (1- \frac{1}{|\mathcal{X}'|})w(OPT(S)) - \frac{1}{|\mathcal{X}'|} [w(Min(S)) + \sum_{x \in \mathcal{X}'} w(\lambda(x))]
  \end{equation}

  $\forall x \in \mathcal{X}'$, $w(P(x)) \le w(\lambda(x))$ since $\lambda(x)$ is the most preferred edge of $x$, so it is obvious that $w(OPT(S)) \le \sum_{x \in \mathcal{X}'} w(\lambda(x))$.\\

  By our assumption,
  \begin{equation*}
  w(Min(S)) < \frac{1}{\sqrt{2}+1}w(OPT(S)) \le \frac{1}{\sqrt{2}+1} \sum_{x \in \mathcal{X}'} w(\lambda(x))
\end{equation*}

  Thus, putting this together with Inequality~\ref{rsd-eq3}, we obtain that,
  \begin{align*}
  \frac{1}{|\mathcal{X}'|} \sum_{x \in \mathcal{X}'} w(OPT(R(x)))
  &\ge w(OPT(S)) - \frac{1}{|\mathcal{X}'|} (2 + \frac{1}{\sqrt{2} + 1}) \sum_{x \in \mathcal{X}'} w(\lambda(x)))\\
  &= w(OPT(S)) - \frac{\sqrt{2} + 1}{|\mathcal{X}'|} \sum_{x \in \mathcal{X}'} w(\lambda(x))
  \end{align*}
\end{proof}

\subsection*{Partial One-sided Ordinal Preferences}
  In this section, we consider the case when we are given even less information than in the previous one, i.e., only partial preferences. 
We begin by establishing the following easy result for the completely random algorithm.

\setcounter{algocf}{7}

\vspace{3mm}
\begin{algorithm}[H]
\caption{Random Algorithm for Perfect Bipartite Matching.}
\label{random}
  Initialize $M =  \emptyset $, $G=(\mathcal{X}, \mathcal{Y}, E)$ \;
   \While {$E \ne \emptyset$} {
     Pick an edge $e = (x,y)$ from $E$ uniformly at random and add it to $M$ \;
     Remove $x$, $y$, and all edges containing $x$ or $y$ from $G$ \;
    }
  \textbf{Final Output:} Return $M$.
\end{algorithm}
\vspace{3mm}

\setcounter{algocf}{1}

\begin{lemma}
\label{random-bound}
  Suppose $G=(\mathcal{X}, \mathcal{Y}, E)$ is a complete bipartite graph on the set of nodes $\mathcal{X}, \mathcal{Y}$ with $|\mathcal{X}| = |\mathcal{Y}| = N$. Then, the expected weight of the random perfect matching returned by Algorithm~\ref{random} for the input $G$ is $E[w(M)] = \frac{1}{N}\sum_{(x, y) \in E}w(x, y)$.
\end{lemma}
This lemma was proved in~\cite{anshelevich2015blind}.\\

 \begin{theorem}
 \label{perfect-random-bound}
 The uniformly random perfect matching is a 3-approximation to the maximum-weight matching.
\end{theorem}


  \begin{proof}
    Suppose $G=(\mathcal{X}, \mathcal{Y}, E)$ is a complete bipartite graph on the set of nodes $\mathcal{X}, \mathcal{Y} \subseteq \mathcal{N}$ with $|\mathcal{X}| = |\mathcal{Y}| = N$. Let $OPT$ be the optimal perfect matching. Suppose $(x, y)$ is an edge in $OPT$. Then for any edge $(a, b) \in E$ , by triangle inequality,
    \begin{equation*}
    w(x,y) \le w(x, b) + w(a, y) + w(a, b)
    \end{equation*}

    Summing up for all $(a, b) \in E$ ,
    \begin{equation*}
    N^2w(x,y) \le N\sum_{b \in \mathcal{Y}} w(x, b) + N\sum_{a \in \mathcal{X}} w(a, y) + \sum_{(a, b) \in E}w(a, b)
    \end{equation*}

    Summing up for all $(x, y) \in OPT$,
    \begin{equation*}
    N^2w(OPT) \le N\sum_{(a, b) \in E}w(a, b) + N\sum_{(a, b) \in E}w(a, b) + N\sum_{(a, b) \in E}w((a, b) = 3N\sum_{(a, b) \in E}w(a, b)
    \end{equation*}

    Let $M$ be the matching returned by Algorithm~\ref{random}. Then, by Lemma~\ref{random-bound},
    \begin{equation*}
    E[w(M)] = \frac{1}{N}\sum_{(a, b) \in E}w(a, b) \ge \frac{1}{3}w(OPT)
    \end{equation*}
  \end{proof}

\vspace{3mm}
\begin{algorithm}[H]
\caption{Algorithm for Perfect Matching given partial one-sided ordering.}
\label{onesided-rsd-partial}
  Run Algorithm~\ref{onesided-rsd}, stop when $|M| = \alpha N$, then form random matches until all agents are matched. Return $M$.
\end{algorithm}
\vspace{3mm}

\begin{theorem}
\label{onesided-rsd-partial-bound}
Suppose $G=(\mathcal{X}, \mathcal{Y}, E)$ is a complete bipartite graph on the set of nodes $\mathcal{X}, \mathcal{Y}$ with $|\mathcal{X}| = |\mathcal{Y}| = N$. There is a strict preference ordering $P_x$ over the agents in $\mathcal{Y}$ for each agent $x\in\mathcal{X}$. We are only given top $\alpha N$ agents in $P_x$ in order. Then, the expected weight of the perfect matching $M$ returned by Algorithm~\ref{onesided-rsd-partial} is $E[w(M)] \ge \frac{1}{3 - (2 - \sqrt{2})\alpha}w(OPT(G))$, as shown in Figure~\ref{fig:plot_figures}.
\end{theorem}

\begin{proof}
  We use the same notation as in the proof of Theorem~\ref{onesided-rsd-bound}. We would like to apply our main technical result (Lemma~\ref{onesided-structural}) to analyze this algorithm. Define $Alg_i(S)$ as the expected weight of chosen edge in round $i$ of RSD on any subgraph $S$. For any bipartite graph $S$, let $Rand(S)$ denote the expected weight of the perfect matching returned by Algorithm~\ref{random}, and $Avg(S)$ denote the average weight of edges in $S$.

We begin by bounding $w(OPT(G))$ by the sum of expected weights of chosen edges in RSD, and the weight of the remaining subgraph.
\begin{lemma}
\label{onesided-partial-structure2}
  Let $L(G,\ell)$ be the subgraph of $G$ after $\ell$ rounds of RSD, which has $N-\ell$ nodes both in $\mathcal{X}$ and $\mathcal{Y}$. Note that $L(G,\ell)$ is a random variable. Then we have that:
  \begin{equation*}
  w(OPT(G)) \le (\sqrt{2}+1)\sum_{i=1}^{\ell}Alg_i(G) + 3E[Rand(L(G,\ell))]
  \end{equation*}
\end{lemma}

\begin{proof}
We prove this by induction on $\ell$. For the Base Case, when $\ell=0$, then this simply reduces to Theorem \ref{random-bound}. Now assume by the inductive hypothesis that, $\forall x \in \mathcal{X}$,
  \begin{equation*}
    w(OPT(R(G,x))) \le (\sqrt{2}+1)\sum_{i=1}^{\ell - 1}Alg_i(R(G,x)) + 3E[Rand(L(R(G,x),\ell-1))]
  \end{equation*}

  If Case 1 in Lemma~\ref{onesided-structural} holds for $G$, then
  \begin{align*}
    w(OPT(G)) &\le \frac{\sqrt{2}+1}{|\mathcal{X}|}\sum_{x \in \mathcal{X}}w(\lambda(G,x)) + \frac{1}{|\mathcal{X}|} \sum_{x \in \mathcal{X}} w(OPT(R(G, x)))\\
    &\le  (\sqrt{2}+1)\sum_{i=1}^{\ell}Alg_i(G) + 3E[Rand(L(G,\ell))]
  \end{align*}

  The last inequality is simply because of the inductive hypothesis, and the fact that $E[Rand(L(G,\ell))] = \frac{1}{|\mathcal{X}|}\sum_{x\in \mathcal{X}}E[Rand(L(R(G,x),\ell-1))]$. If instead Case 2 in  Lemma~\ref{onesided-structural} holds for $G$, then
  \begin{equation*}
    w(OPT(G)) \le (\sqrt{2}+1)w(Min(G))
  \end{equation*}
  Let's consider a perfect matching on $G$ generated by running RSD for $\ell$ rounds, and then obtaining the minimum weight matching for the remaining subgraph. By the definition of $Min(G)$, the weight of the matching described above is no less than $Min(G)$:
  \begin{align*}
    w(OPT(G)) &\le (\sqrt{2}+1)w(Min(G))\\
              &\le (\sqrt{2}+1)\sum_{i=1}^{\ell}Alg_i(G) + (\sqrt{2}+1) E[w(Min(L(G,\ell)))]\\
              &\le (\sqrt{2}+1)\sum_{i=1}^{\ell}Alg_i(G) + (\sqrt{2}+1) E[Rand(L(G,\ell))]\\
              &\le (\sqrt{2}+1)\sum_{i=1}^{\ell}Alg_i(G) + 3E[Rand(L(G,\ell))]
  \end{align*}
\end{proof}

To finish the proof of the theorem, we need to be able to compare $Rand(L(G,\ell))$ and $Alg_i$. After all, if the random part of our matching is much larger in weight than the RSD part, then the random part will dominate, resulting in only a 3 approximation. Fortunately, it is not hard to see the following lemma. 
Let $G' = L(G,\alpha N)$ be a random variable representing the graph obtained by running RSD on $G$ for $\alpha N$ rounds, which we can always do if we are given the top $\alpha N$ preferences of every agent.

  \begin{lemma}
  \label{alg_avg}
    $\forall i \le \alpha N$, $Alg_i(G)$ is heavier than the expected average edge weight in $G'$, i.e., $Alg_i(G) \ge E[Avg(G']$.
  \end{lemma}

  \begin{proof}
  First notice that $Alg_1(G) \ge Alg_2(G)$. This is true because:
  \begin{align*}
  Alg_2(G) &= \frac{1}{|\mathcal{X}|}\sum_{x\in\mathcal{X}}\frac{1}{|\mathcal{X}|-1}\sum_{y\in\mathcal{X}-x}w(\lambda(R(G,x),y))\\
           &\le \frac{1}{|\mathcal{X}|}\sum_{x\in\mathcal{X}}\frac{1}{|\mathcal{X}|-1}\sum_{y\in\mathcal{X}-x}w(\lambda(G,y))\\
           &= \frac{|\mathcal{X}|-1}{|\mathcal{X}|(|\mathcal{X}|-1)}\sum_{y\in\mathcal{X}}w(\lambda(G,y))\\
           &= Alg_1(G)
  \end{align*}
  The inequality above is simply because the best edge leaving $y$ in a smaller graph $R(G,x)$ is at most the best edge leaving it in a larger graph $G$. By the same argument, we know that $Alg_i(G)\geq Alg_{i+1}(G)$ for all $i$.

  Now consider an arbitrary complete graph $S=(\mathcal{X}',\mathcal{Y}',E')$ with $|\mathcal{X}'|=|\mathcal{Y}'|$. One way to think of $Avg(S)$ is as an expected value of the following randomized algorithm: take a node $x$ in $\mathcal{X}'$ uniformly at random, and then take a random edge leaving that node, and return its weight. The expected value returned by this algorithm is exactly the expected weight of an edge in $S$ taken uniformly at random, i.e., exactly $Avg(S)$. Compare this algorithm with the performance of RSD; RSD does exactly the same thing in the first round, but chooses the best edge coming out of $x$ instead of a random edge. Therefore, the first round of RSD on any graph always performs better than the average edge weight. In particular, this is true for every instantiation of the graph $G'$, and thus $Alg_{\alpha N +1}(G)\geq E[Avg(G')]$. This concludes the proof.
  \end{proof}

  Finally, let's finish the proof of Theorem \ref{onesided-rsd-partial-bound}. By Lemma~\ref{onesided-partial-structure2},

  \begin{align*}
    w(OPT(G)) &\le (\sqrt{2}+1)\sum_{i=1}^{\alpha N}Alg_i(G) + 3E[Rand(G')]\\
    &= (\sqrt{2}+1)\sum_{i=1}^{\alpha N}Alg_i(G) + 3(1 - \alpha) N \times E[Avg(G')],
  \end{align*}

  By Lemma~\ref{alg_avg},

  \begin{equation*}
     \sum_{i=1}^{\alpha N}Alg_i(G) \ge  \alpha N \times E[Avg(G')],
  \end{equation*}

and thus,

  \begin{align*}
    w(OPT(G)) &\le (3 - (2 - \sqrt{2})\alpha)(\sum_{i=1}^{\alpha N}Alg_i(G) + (1-\alpha)N \times E[Avg(G')])\\
              &= (3 - (2 - \sqrt{2})\alpha)(\sum_{i=1}^{\alpha N}Alg_i(G) + E[Rand(G')])
  \end{align*}

  Note that $\sum_{i=1}^{\alpha N}Alg_i(G) + E[Rand(G')]$ is the expected weight of $M$, which completes the proof:
  \begin{equation*}
    w(OPT(G)) \le (3 - (2 - \sqrt{2})\alpha)E[w(M)].
  \end{equation*}
\end{proof}

\section{Two-sided Ordinal Preferences}
\label{section-twosided}
For two-sided preferences, we give separate algorithms for the cases when $\alpha\geq \frac{1}{2}$ and when $\alpha\leq\frac{1}{2}$, as these require somewhat different techniques. 


\vskip 5pt\noindent{$\bm{\alpha\geq\frac{1}{2}}$~~} While for the case when $\alpha < \frac{1}{2}$ new techniques are necessary to obtain a good approximation, the approach for the case when $\alpha\geq \frac{1}{2}$ is essentially the same as the one used in \cite{anshelevich2015blind}. 
We adopt this approach to deal with bipartite graphs and with partial preferences, giving us a 1.8-approximation 
for $\alpha=1$. To do this, we re-state the definition of Undominated Edges from \cite{anshelevich2015blind}, and a standard greedy algorithm for forming a matching of size $k$.

%
%
%

\begin{definition}
  (Undominated Edges) Given a set $E$ of edges, $(x, y) \in E$ is said to be an undominated edge if for all $(x, a)$ and $(y, b)$ in $E$, $w(x,y) \ge w(x,a)$ and $w(x,y) \ge w(y,b)$.
\end{definition}

Note that an undominated edge must always exist: either there are two nodes $x$ and $y$ such that they are each other's top preferences (and so $(x,y)$ is undominated), or there is a cycle $x_1,x_2,\ldots$ in which $x_{i+1}$ is the top preference of $x_i$, in which case all edges in the cycle must be the same weight, and thus all edges in the cycle are undominated. This also gives us an algorithm for determining if an edge $(x,y)$ is undominated: either $x$ and $y$ prefer each other over all other agents, or it is part of such a cycle of top preferences.

\begin{lemma}
\label{undominated-weight-bound}
  Given an edge set $E$ of a complete bipartite graph $G = (\mathcal{X}, \mathcal{Y}, E)$, the weight of any undominated edge is at least one third as much as the weight of any other edge in $E$, i.e., if $e=(x, y)$ is an undominated edge in $E$, that $x \in \mathcal{X}$,  $y \in \mathcal{Y}$, then for any $(a, b) \in E$,  $a \in \mathcal{X}$,  $b \in \mathcal{Y}$, $w(x, y) \ge \frac{1}{3}w(a,b)$.\\
\end{lemma}

\begin{proof}
  Since $e=(x, y)$ is an undominated edge, $w(x, y) \ge w(x, b)$ and $w(x, y) \ge w(a, y)$. By triangle inequality, we know that $w(a, b) \le w(x,y) + w(x, b) + w(a, y) \le 3w(x, y)$.
\end{proof}


\vspace{3mm}
\begin{algorithm}[htb]
\caption{Greedy Algorithm for Max $k$-Matching of two-sided ordering.}
\label{twosided-greedy}
  Initialize $M =  \emptyset $, $E$ is the valid set of edges initialized to the complete bipartite graph $G$ \;
   \While {$E \ne \emptyset$} {
     Pick an undominated edge $e=(x,y)$ from E and add it to $M$  \;
     Remove $x$, $y$, and all edges containing $x$ or $y$ from $E$ \;
     \If { $|M| = k$ }  {
     	break \;
      }
    }
  \textbf{Final Output:} Return $M$.
\end{algorithm}
\vspace{3mm}

Before stating the full algorithm for the case when $\alpha\geq\frac{1}{2}$, we mention two lemmas which will be useful to establish its approximation ratio. These lemmas are essentially the same as the similar ones from ~\cite{anshelevich2015blind}, except that we must adjust all the factors to deal with bipartite graphs, while ~\cite{anshelevich2015blind} considered only non-bipartite graphs. The basic analysis techniques remain the same, however, and we only provide proofs of these lemmas for completeness.

\begin{lemma}
\label{twosided-greedy-k}
  Suppose $G=(\mathcal{X}, \mathcal{Y}, E)$ is a complete bipartite graph on the set of nodes $\mathcal{X}, \mathcal{Y}$ with $|\mathcal{X}| = |\mathcal{Y}| = N$. Given $k = \gamma N$, the performance of the greedy $k$-matching returned by Algorithm~\ref{twosided-greedy} with respect to the optimal perfect matching OPT is given by $\frac{3-2\gamma}{\gamma}$.
\end{lemma}

\begin{proof}
  The analysis here is essentially identical to that of a similar lemma in ~\cite{anshelevich2015blind}, except that Lemma~\ref{undominated-weight-bound} gives a ratio of 3 instead of 2 between any edge and an undominated edge for bipartite graphs. We include the whole analysis of the framework for completeness.

  Let $M$ be the greedy $k$-matching, and $M^*$ be the optimal perfect matching. We show the claim by charging every edge in $M^*$ to one or more edges in the greedy matching $M$. Consider any edge $e^* = (a, b)$ in $M^*$, the edge must belong to one of the following two types.
  \begin{enumerate}
    \item (Type I) Some edges consisting of $a$ or $b$ (both $a$ and $b$) are present in $M$.
    \item (Type II) No edge in $M$ has $a$ or $b$ as an endpoint.
  \end{enumerate}
  Suppose that $M^*$ contains $m_1$ Type I edges, and $m_2$ Type II edges. We know that $m_1 + m_2 = N$. Let $T \subset M$ denote the heaviest $\frac{m_1}{2}$ edges in $M$. Initialize $U$ as all the edges in $M$. We describe our charging algorithm in three phases.

  (First Phase) We can charge all Type I edges in $M^*$ to the edges in  $T$, so that $\sum_{e \in T}s_e w_e \ge \sum_{e \in TypeI (M^*)}w_e$, $s_e \le 2$. We charge the edges as follows: Repeat until $U$ contains no Type I edge: pick a type I edge $e^*=(a,b) $ from $U$. Suppose that $e = (a,c)$ is the first edge containing either a or b that was added to $M$, Since $w_e \ge w_e^*$, charge $e^*$ to $e$, increase $s_e$ by one and remove $e^*$ from U.
   In the end, all the edges that are charged in $M$ have $s_e \le 2$, and $\sum_{e}s_e = m_1$ . We can transfer the slots to the heaviest $\frac{m_1}{2}$ edges in $M$, each has $s_e \le 2$. Keep transfering the slots to the heaviest $\frac{m_1}{\mu}$ edges in $M$, so that each edge has $s_e \le \mu$.

  (Second Phase) Repeat until  $s_e = \mu$ for all $e \in M \backslash T$ or until $U$ is empty: pick any arbitrary edge $e^*$ from $U$ and the smallest edge $e \in M \backslash T$ such that $s_e < \mu$ By Lemma~\ref{undominated-weight-bound} $w_e* \le 3w_e$, charge $e*$ using three slots of $e$, transfer slots to the heaviest edges $e \in M \backslash T$ such that $s_e < \mu$. So   $e^*$ is charged  by three slot from edges in $ M \backslash T$.

  At the end of the second phase, $|U| = \max(0, m_2 - (k - \frac{m_1}{\mu} ) \times \frac{\mu}{3})$.\\

  (Third Phase) Repeat until $U$ is empty: pick any arbitrary edge $e^*$ from $U$. Since $w_e^* \le 3w_e$ for all $e \in M$, charge $e^*$ uniformly to all edges in $M$, i.e., increase $s_e$ by $\frac{3}{k}$ for every $e \in M$ and remove $e^*$ from $U$. \\

  At the end of the third phase, for every $e \in M$,
  \begin{equation*}
  s_e \le \mu + \frac{3}{k} \max(0, m_2 -  (k - \frac{m_1}{\mu} ) \times \frac{\mu}{3})
  \end{equation*}

  Because $m_1 + m_2 = N$,
  \begin{equation*}
  s_e \le \max(\mu, \frac{2m_2 +N}{k})
  \end{equation*}

  Type II edges don't share nodes with any of the $k$ edges in $M$, so $m_2 + k \le N$,
  \begin{equation*}
  s_e \le \max(\mu, \frac{3N - 2k}{k})
  \end{equation*}
  \begin{equation*}
  s_e \le  max(\mu, \frac{3-2\gamma}{\gamma})
  \end{equation*}

  Let $\mu = \frac{3-2\gamma}{\gamma}$,
  \begin{equation*}
  s_e \le \frac{3-2\gamma}{\gamma}
  \end{equation*}
\end{proof}

\begin{lemma}
\label{upper-bound}

  Let $G_T = (\mathcal{X}_T, \mathcal{Y}_T, E_T )$ be a complete bipartite subgraph on the set of nodes $\mathcal{X}_T \subseteq \mathcal{X}$, $\mathcal{Y}_T \subseteq \mathcal{Y}$, with $|\mathcal{X}_T| = |\mathcal{Y}_T| = n$, and let $M$ be any perfect matching on $G  = (\mathcal{X}, \mathcal{Y}, E )$. Then, the following is an upper bound on the weight of $M$,
  \begin{equation*}
  nw(M) \le (2 + \frac{N}{n}) \sum_{\substack{x \in \mathcal{X}_T \\ y \in  \mathcal{Y}_T}} w(x, y) +  \sum_{\substack{x \in \mathcal{X}_T \\ y \in  \mathcal{Y} \backslash \mathcal{Y}_T}} w(x, y)  + \sum_{\substack{x \in  \mathcal{X} \backslash \mathcal{X}_T \\ y \in  \mathcal{Y}_T}} w(x, y)
  \end{equation*}

\end{lemma}

\begin{proof}

  For $e = (x, y) \in M, e'=(a, b) \in E_T$, by triangle inequality,
  \begin{equation*}
  w(a, y) + w(a, b) + w(x, b) \ge w(x, y)
\end{equation*}

  Sum up for all $(a, b) \in E_T$,
  \begin{equation*}
  n \sum_{a \in \mathcal{X}_T} w(a, y) + \sum_{\substack{a \in \mathcal{X}_T \\ b \in  \mathcal{Y}_T}} w(a, b) + n \sum_{b \in \mathcal{Y}_T} w(x, b) \ge n^2 w(x, y)
  \end{equation*}

  Sum up for all $(x, y) \in M$,
  \begin{equation*}
  n \sum_{\substack{a \in \mathcal{X}_T \\ y \in \mathcal{Y} }} w(a, y) + N \sum_{\substack{a \in \mathcal{X}_T \\ b \in  \mathcal{Y}_T}} w(a, b) + n \sum_{\substack{b \in \mathcal{Y}_T \\ x \in \mathcal{X} }} w(x, b) \ge n^2 w(M)
  \end{equation*}

  Because $\mathcal{Y} =  \mathcal{Y}_T \cup  \{ \mathcal{Y} \backslash \mathcal{Y}_T \} $, and $\mathcal{X} =  \mathcal{X}_T \cup  \{ \mathcal{X} \backslash \mathcal{X}_T\} $,
  \begin{equation*}
  n (\sum_{\substack{a \in \mathcal{X}_T \\ y \in \mathcal{Y}_T }} w(a, y)  + \sum_{\substack{a \in \mathcal{X}_T \\ y \in \mathcal{Y} \backslash \mathcal{Y}_T  }} w(a, y) ) + N \sum_{\substack{a \in \mathcal{X}_T \\ b \in  \mathcal{Y}_T}} w(a, b) + n (\sum_{\substack{b \in \mathcal{Y}_T \\ x \in \mathcal{X}_T }} w(x, b) + \sum_{\substack{b \in \mathcal{Y}_T \\ x \in \mathcal{X} \backslash \mathcal{X}_T }} w(x, b) ) \ge n^2 w(M)
  \end{equation*}

  Replace $a$ with $x$, and $b$ with $y$,
  \begin{equation*}
    2n \sum_{\substack{x \in \mathcal{X}_T \\ y \in  \mathcal{Y}_T}} w(x, y) +  n \sum_{\substack{x \in \mathcal{X}_T \\ y \in  \mathcal{Y} \backslash \mathcal{Y}_T}} w(x, y)  + n \sum_{\substack{x \in  \mathcal{X} \backslash \mathcal{X}_T \\ y \in  \mathcal{Y}_T}} w(x, y) + N \sum_{\substack{x \in \mathcal{X}_T \\ y \in  \mathcal{Y}_T}} w(x, y) \ge n^2 w(M)
  \end{equation*}
  \begin{equation*}
  nw(M) \le (2 + \frac{N}{n}) \sum_{\substack{x \in \mathcal{X}_T \\ y \in  \mathcal{Y}_T}} w(x, y) +  \sum_{\substack{x \in \mathcal{X}_T \\ y \in  \mathcal{Y} \backslash \mathcal{Y}_T}} w(x, y)  + \sum_{\substack{x \in  \mathcal{X} \backslash \mathcal{X}_T \\ y \in  \mathcal{Y}_T}} w(x, y)
  \end{equation*}

\end{proof}

%

We can now state the algorithm for $\alpha \ge \frac{1}{2}$. The algorithm is a mix of greedy and random algorithms: for graph $G=(\mathcal{X}, \mathcal{Y}, E)$, given top $\alpha N$ of $P(\mathcal{X})$ and top $\alpha N$ of $P(\mathcal{Y})$, run Algorithm~\ref{twosided-greedy} on $k = \alpha N$, to obtain the matching $M_0$. This is possible using the preference we are given. One method we could do at this point is to form a random matching on the rest of the agents. However, this will not form a good approximation, as there are examples when all the high-weight edges are between nodes matched in $M_0$ and nodes which are unmatched. Another method is to randomly choose some matched nodes from $M_0$, make then unmatched, and form a random bipartite matching between all the agents which were not matched in $M_0$, and the nodes which we chose from $M_0$ to become unmatched. This second method is likely to add high-weight edges between nodes in $M_0$ and nodes outside of it to our matching. Mixing over these two methods actually returns a high-weight matching in expectation.

Note that for $\alpha>\frac{3}{4}$ this algorithm does not seem to provide better guarantees than for $\alpha=\frac{3}{4}$. Because of this, for $\alpha>\frac{3}{4}$, we simply run the same algorithm for $\alpha=\frac{3}{4}$

\vspace{3mm}
\begin{algorithm}[htb]
\caption{Algorithm for two-sided matching with partial ordinal information ($ \frac{1}{2} \le \alpha \le \frac{3}{4}$).}
\label{partial-twosided-alg}
  \SetKwInOut{Input}{Input}
    	\SetKwInOut{Output}{Output}
  \Input{$\mathcal{X}, \mathcal{Y}$, top $\alpha N$ of $P(\mathcal{X})$, top $\alpha N$ of $P(\mathcal{Y})$}
  \Output{ Perfect Bipartite Matching M }
  Initialize $E$ to be complete bipartite graph on $\mathcal{X}, \mathcal{Y}$, and $M_1 = M_2 = \emptyset $ \;
  Let $M_0 $ be the output returned by Algorithm~\ref{twosided-greedy} for $E$, $k = \alpha N$ \;
  Let $\mathcal{X}_T$ be the set of nodes in  $\mathcal{X}$ matched in $M_0$, $\mathcal{Y}_T$ be the set of nodes in $\mathcal{Y}$ matched in $M_0$, and $T$ be the complete bipartite graph on $\mathcal{X}_T, \mathcal{Y}_T$\ \;
  Let $\mathcal{X}_B = \mathcal{X} \backslash \mathcal{X}_T$, $\mathcal{Y}_B = \mathcal{Y} \backslash \mathcal{Y}_T$, and $B$ be the complete bipartite graph on $\mathcal{X}_B, \mathcal{Y}_B$\;

  \textbf{First Algorithm};\\
  $M_1 = M_0 \cup$ (The perfect matching output by Algorithm~\ref{random}  on $B$)\;

  \textbf{Second Algorithm};\\
  Choose $(2\alpha - 1)N$ edges from $M_0$ uniformly at random and add them to $M_2$ \;
  Let $X_A$ be the set of nodes in $ \mathcal{X}_T$ and not in $M_2$,  $Y_A$ be the set of nodes in $ \mathcal{Y}_T$ and not in  $M_2$\;
  Let $E_{AB}$ be the edges of the complete bipartite graph $(X_A, \mathcal{Y}_B)$ and $E_{AB}'$ be the edges of the complete bipartite graph $(\mathcal{X}_B, Y_A)$ \;
  Run random bipartite matching on the set of edges in $E_{AB}$ and $E_{AB}'$ separately to obtain perfect bipartite matchings and add the edges returned by the algorithm to $M_2$\;
  \textbf{Final Output:} Return $M_1$ with probability $\frac{3-2\alpha}{3-\alpha}$ and $M_2$ with probability $\frac{\alpha}{3-\alpha}$.
\end{algorithm}
\vspace{3mm}

\noindent Note that for $\alpha>\frac{3}{4}$ this algorithm does not seem to provide better guarantees than for $\alpha=\frac{3}{4}$. Because of this, for $\alpha>\frac{3}{4}$, we simply run the same algorithm for $\alpha=\frac{3}{4}$.

\begin{theorem}
  \label{partial-twosided}
  Algorithm \ref{partial-twosided-alg} returns a $\frac{(3 - 2\alpha)(3 - \alpha)}{2\alpha^2-3\alpha+3}$-approximation to the maximum-weight perfect matching given two-sided ordering when $ \frac{1}{2} \le \alpha \le \frac{3}{4}$. 
\end{theorem}

\begin{proof}
  $ |\mathcal{X}_T| = |\mathcal{Y}_T| = \alpha N$, $ |\mathcal{X}_B| = |\mathcal{Y}_B| = (1 - \alpha)N$.\\
  By Lemma~\ref{twosided-greedy-k}, $w(M_0) \ge \frac{\alpha}{3 - 2\alpha}OPT$. By Lemma~\ref{random-bound}, the perfect matching output by Algorithm~\ref{random}  on $B$ has expected weight at least $\frac{1}{(1-\alpha)N}w(B)$. Therefore,

  \begin{equation*}
  E[w(M_1)] \ge \frac{\alpha}{3 - 2\alpha}OPT + \frac{1}{(1 - \alpha)N}w(B)
  \end{equation*}

  Because $|X_A| = |Y_A| = (1-\alpha)N$, and they are leftover nodes after $(2\alpha - 1)N$ nodes are chosen uniformly at random from $M_0$,
  \begin{equation*}
  E[w(E_{AB}) + w(E_{AB}')] = \frac{1- \alpha}{\alpha} w(T, B).
  \end{equation*}

  Recall that $w(T,B)$ is the total weight of all edges between $T$ and $B$. Let $M_{AB}$ be a random bipartite matching formed on edges $E_{AB}$ and $E'_{AB}$. By Lemma~\ref{random-bound},
  \begin{align*}
  E[w(M_{AB})] &= \frac{1}{(1 - \alpha)N}E[w(E_{AB})] + \frac{1}{(1 - \alpha)N}E[w(E_{AB}')] \\
        &= \frac{1}{(1 - \alpha)N}E[w(E_{AB}) + w(E_{AB}')]\\
        &= \frac{1}{ \alpha N} w(T, B)
  \end{align*}

  By Lemma~\ref{upper-bound}, with $M=OPT, T = B, n=(1 - \alpha)N$:
  \begin{equation*}
  (1 - \alpha)N w(OPT) \le (2 + \frac{1}{1 - \alpha}) w(B) + w(T, B)
  \end{equation*}

  \begin{align*}
   E[w(M_{AB})] &= \frac{1}{\alpha N} w(T, B)\\
                &\ge \frac{1}{\alpha N}  ((1 - \alpha)N w(OPT) - \frac{3 - 2\alpha}{1-\alpha}w(B) )
 \end{align*}

  $M_2$ contains $\frac{2\alpha - 1}{\alpha}$ fraction of edges randomly chosen from $M_0$, together with $M_{AB}$:
  \begin{align*}
   E[w(M_2)] &= \frac{2\alpha - 1}{\alpha} \times \frac{\alpha}{3 - 2\alpha} w(OPT) +  E[w(M_{AB})]\\
   &\ge \frac{2\alpha - 1}{3 - 2\alpha}w(OPT) + \frac{1}{\alpha N}  ((1 - \alpha)N w(OPT) - \frac{3 - 2\alpha}{1-\alpha}w(B) )\\
   &= \frac{4\alpha^2 - 6\alpha+3}{\alpha(3 - 2\alpha)}w(OPT) - \frac{3 - 2\alpha}{\alpha (1-\alpha) N}w(B)
  \end{align*}

  Return $M_1$ with probability $\frac{3 - 2\alpha}{3 - \alpha}$ and $M_2$ with probability $\frac{\alpha}{3 - \alpha}$. Then, the expected weight of our final matching is
  \begin{equation*}
   \frac{3 - 2\alpha}{3 - \alpha}E[w(M_1)] + \frac{\alpha}{3 - \alpha}E[w(M_2)]
   \ge \frac{2\alpha^2 - 3\alpha + 3}{(3 - 2\alpha)(3 - \alpha)}w(OPT).
  \end{equation*}
\end{proof}

\vskip 5pt\noindent{$\bm{\alpha\leq\frac{1}{2}}$~~}
Unlike the case for $\alpha\geq\frac{1}{2}$, this case requires different techniques than in \cite{anshelevich2015blind}. While the techniques above would still work, they will not give us a bound as good as the one we form below. The idea in this section is to do something similar to our one-sided algorithm for partial preferences: run the greedy algorithm for a while, and then switch to random. Unfortunately, if we simply run the greedy Algorithm \ref{twosided-greedy} and then switch to random, this will not form a good approximation. The reason why this is true is that an undominated edge which is picked by the greedy algorithm may be much worse than the average weight of an edge, and so the approximation factor of the random algorithm will dominate, giving only a 3-approximation. Even taking an undominated edge uniformly at random has this problem. We can fix this, however, by picking each undominated edge with an appropriate probability, as described below. Such an algorithm results in matchings which are guaranteed to be better than either RSD or Random, thus allowing us to prove the result.

\vspace{3mm}
\begin{algorithm}[htb]
\caption{Algorithm for two-sided matching with partial ordinal information ($  0 \le \alpha \le \frac{1}{2}$).}
\label{partial-twosided-alg2}
\SetKwInOut{Input}{Input}
    \SetKwInOut{Output}{Output}
\Input{$\mathcal{X}, \mathcal{Y}$, top $\alpha N$ of $P(\mathcal{X})$ and $P(\mathcal{Y})$}
Initialize $M =  \emptyset $, $G = (\mathcal{X}, \mathcal{Y}, E)$ \;
  \While {$E \ne \emptyset$} {
    Pick an agent $x$ uniformly at random from $\mathcal{X} $ \;
    Let $y$ denote $x$'s most preferred agent in $\mathcal{Y} $ \;
    $x_1 \gets x$, $y_1 \gets y$,  $c \gets y_1$\;
    \While {$(x_1, y_1)$ is not an undominated edge} {
      \eIf {$c$ = $y_1$} {
         $x_1 \gets y_1$'s most preferred agent in $\mathcal{X}$ \;
         $c \gets x_1$\;
      } {
         $y_1 \gets x_1$'s most preferred agent in $\mathcal{Y}$ \;
         $c \gets y_1$\;
      }
    }
    Take  $(x_1, y_1)$ from $E$ and add it to $M$ \;
    Remove $x_1$, $y_1$, and all edges containing $x_1$ or $y_1$ from the graph $G$ \;
    \If {$|M| = \alpha N$} {
      break\;
    }
   }
   Run Algorithm~\ref{random} for the remaining graph $G$, add the edges returned by the algorithm to $M$.
 \textbf{Final Output:} Return $M$.
\end{algorithm}
\vspace{3mm}

This algorithm guarantees that an undominated edge is chosen for any $x$ in any bipartite graph $G$. Now, before we reach an undominated edge, the weights of edges are non-decreasing in the order they are checked. 
Thus whenever a node $x$ is picked, the algorithm adds an undominated edge $(x_1,y_1)$ to the matching which is guaranteed to have higher weight than all edges leaving $x$. Note that it is not possible to apply this algorithm to one-sided matching because the preferences of agents in $\mathcal{Y}$ are not given, and thus we cannot detect which edges are undominated.

\begin{theorem}
  \label{partial-twosided-2}
  Algorithm \ref{partial-twosided-alg2} returns a $(3 - \alpha)$-approximation to the maximum-weight perfect matching given two-sided ordering when $ 0 \le \alpha \le \frac{1}{2}$. 
\end{theorem}

\begin{proof}
We use a similar method and the same notation as in Section~\ref{section-onesided} to proof this theorem. Essentially, because we are always picking undominated edges, we can form a linear interpolation between a factor of 2 and a factor of 3 for random matching, instead of between factors $\sqrt{2}+1$ and 3 as for one-sided preferences. The reason why we are able to form such an interpolation is entirely because of the probabilities with which we choose the undominated edges; if we simply chose arbitrary undominated edges or choose them uniformly at random, then there are examples where the random edge weights will dominate and result in a poor approximation, since undominated edges are only guaranteed to be within a factor of 3 of the average edge weight.

Besides those was used in the proof of Theorem~\ref{onesided-rsd-bound}, we introduce some new notation. Suppose that Algorithm~\ref{partial-twosided-alg2} picks $x \in \mathcal{X}'$,  and end up with an undominated edge $(x_1, y_1)$. Let $\lambda_D(S, x)$ denote the undominated edge picked by the algorithm for $x$ in graph $S$, $\lambda_D(S, x) = (x_1, y_1) = \lambda(S,x_1)$ in this case. And let $R_D(S, x)$ denote the remaining graph after removing $\lambda_D(S, x)$ and the edges connected to both vertexes of $\lambda_D(S, x)$.

We start with a lemma to bound the maximum weight matching,
\begin{lemma}
  \label{twosided-partial-structural}
  For any given subgraph $S = (\mathcal{X}', \mathcal{Y}', E')$,  $w(OPT(S)) \le \frac{1}{|\mathcal{X}'|} \sum_{x \in \mathcal{X}'} w(OPT(R_D(S, x))) + \frac{2}{|\mathcal{X}'|}\sum_{x \in \mathcal{X}'} w(\lambda_D(S, x))$. \\
\end{lemma}

\begin{proof}
  Using the same notation as in the proofs of Theorems \ref{onesided-rsd-bound} and \ref{onesided-rsd-partial-bound}, suppose that Algorithm~\ref{partial-twosided-alg2} picks $x \in \mathcal{X}'$,  and end up with an undominated edge $(x_1, y_1)$. Then $OPT(R_D(S, x)) = OPT(R(S, x_1))$ is at least as good as the matching obtained by removing $P(x_1)$ and $\bar{P}(x_1)$, and adding $D(x_1)$ to $OPT(S)$ (the rest stay the same):
\begin{align*}
w(OPT(R_D(S, x))) &\ge w(OPT(S)) - w(P(x_1)) - w(\bar{P}(x_1)) + w(D(x_1))\\
               &\ge w(OPT(S)) - w(P(x_1)) - w(\bar{P}(x_1))
\end{align*}

Because $ \lambda_D(S, x)$ is an undominated edge, $w(\lambda_D(S, x)) \ge P(x_1)$, $w( \lambda_D(S, x)) \ge \bar{P}(x_1)$,
\begin{equation*}
w(OPT(R_D(S, x))) \ge w(OPT(S)) - 2 w(\lambda_D(S, x))
\end{equation*}

Summing up for all $x$ in $\mathcal{X}'$,
\begin{equation*}
  \sum_{x \in \mathcal{X}'} w(OPT(R_D(S, x))) \ge |\mathcal{X}'| w(OPT(S)) - 2 \sum_{x \in \mathcal{X}'} w(\lambda_D(S, x))
\end{equation*}

\begin{equation*}
  w(OPT(S)) \le \frac{1}{|\mathcal{X}'|} \sum_{x \in \mathcal{X}'} w(OPT(R_D(S, x))) + \frac{2}{|\mathcal{X}'|}\sum_{x \in \mathcal{X}'} w(\lambda_D(S, x)).
\end{equation*}
\end{proof}

Then we bound $w(OPT(G))$ by the sum of expected weights of chosen edges in Algorithm~\ref{partial-twosided-alg2}, and the weight of the remaining subgraph. We still use $Alg_i(S)$ as the expected weight of chosen edge in round $i$, but note that for any $x$, the chosen edge is $\lambda_D(G, x)$ instead of $\lambda(G, x)$ as in Theorem~\ref{onesided-rsd-partial-bound}. By an identical argument as in our Lemma \ref{onesided-partial-structure2}, we have that the following holds:

\begin{equation*}
w(OPT(G)) \le 2\sum_{i=1}^{\ell}Alg_i(G) + 3E[Rand(L(G,\ell))].
\end{equation*}

%
%

We need to prove that a version of Lemma~\ref{alg_avg} still holds for Algorithm~\ref{partial-twosided-alg2}, as the edges are chosen differently from RSD in each step. In other words, we need to be able to compare $Rand(L(G,\ell))$ and $Alg_i$. This is where we need to use the fact that each undominated edge is carefully chosen with a specific probability. Let $G' = L(G,\alpha N)$ be a random variable representing the graph obtained by running our greedy algorithm on $G$ for $\alpha N$ rounds, which we can always do if we are given the top $\alpha N$ preferences of every agent.

  \begin{lemma}
  \label{alg_avg2}
    $\forall i \le \alpha N$, $Alg_i(G)$ is heavier than the expected average edge weight in $G'$, i.e., $Alg_i(G) \ge E[Avg(G']$.
  \end{lemma}

  \begin{proof}
  We must show that $Alg_1(G) \ge Alg_2(G)$. To see this,
  \begin{align*}
  Alg_2(G) &= \frac{1}{|\mathcal{X}|}\sum_{x\in\mathcal{X}}\frac{1}{|\mathcal{X}|-1}\sum_{y\in\mathcal{X}-\lambda_D(G,x)}w(\lambda_D(R_D(G,x),y))\\
           &\le \frac{1}{|\mathcal{X}|}\sum_{x\in\mathcal{X}}\frac{1}{|\mathcal{X}|-1}\sum_{y\in\mathcal{X}-\lambda_D(G,x)}w(\lambda_D(G,y))\\
  \end{align*}
  The inequality above is because the undominated edge found after selecting $y$ and then following the agents' top preferences in a smaller graph $R_D(G,x)$ is at most that in a larger graph $G$.

  Fix some $x\in\mathcal{X}$, and let $(x_1,y_1)$ be the edge $\lambda_D(G,x)$ be the edge added to the matching if $x$ is picked by our algorithm, and thus $x_1$ is the node removed from $\mathcal{X}$. Note that for the case when $x\neq x_1$, we still have that $w(\lambda_D(G,x)) = w(\lambda_D(G,x_1))$, since if $x_1$ is picked by our algorithm, then the undominated edge next to it $(x_1,y_1)$ is immediately returned. Therefore, in the sum above, we can replace $w(\lambda_D(G,x))$ (since $x$ still remains in $\mathcal{X}-\lambda_D(G,x)$) with $w(\lambda_D(G,x_1))$, and thus equivalently make the sum be over $\mathcal{X}-x$ instead of over $\mathcal{X}-\lambda_D(G,x)$.

  \begin{align*}
  Alg_2(G) &= \le \frac{1}{|\mathcal{X}|}\sum_{x\in\mathcal{X}}\frac{1}{|\mathcal{X}|-1}\sum_{y\in\mathcal{X}-x}w(\lambda_D(G,y))\\
           &= \frac{|\mathcal{X}|-1}{|\mathcal{X}|(|\mathcal{X}|-1)}\sum_{y\in\mathcal{X}}w(\lambda_D(G,y))\\
           &= Alg_1(G)
  \end{align*}
  By the same argument, we know that $Alg_i(G)\geq Alg_{i+1}(G)$ for all $i$.


 All that is left is to compare $Alg_{\alpha N+1}(G)$ with $E[Avg(G')]$. We know that the first round of RSD on any graph always performs better than the average edge weight. And for every $x$ that is chosen uniformly at random in the first step of Algorithm~\ref{partial-twosided-alg2}, the weight of final chosen edge $\lambda_D(x)$ is no smaller than $\lambda_(x)$. Therefore, the expected weight of chosen edge in the first round of Algorithm~\ref{partial-twosided-alg2} is no smaller than that of RSD, thus better than the average edge weight, $Alg_{\alpha N+1}(G) \ge E[Avg(G')]$. This concludes the proof.
\end{proof}

Finally, to finish the proof of Theorem \ref{partial-twosided-2}. Similarly to the proof of Theorem \ref{onesided-rsd-partial-bound}, it is easy to show that there is a linear tradeoff from 3 to 2-approximation for $\alpha=0$ to $\alpha=1$, which gives $w(OPT(G)) \le (3 - \alpha)E[w(M)]$, in which $M$ is a random variable representing the matching returned by Algorithm~\ref{partial-twosided-alg2}.
\end{proof}

\section{Total Ordering of Edge Weights}
\label{section-totalorder}
For the setting in which we are given the top $\alpha N^2$ edges of $G$ in order, we prove that for $\alpha=\frac{3}{4}$, we can obtain an approximation of $\frac{5}{3}$ in expectation. For larger $\alpha$, however, more information does not seem to help, and so we simply use the algorithm for $\alpha=\frac{3}{4}$ for any $\alpha>\frac{3}{4}$.



\vspace{3mm}
\begin{algorithm}[H]
\caption{Greedy Algorithm for Max $k$-Matching given the total ordering of edge weights.}
\label{total-order-greedy}
  Initialize $M =  \emptyset $, $E$ is the valid set of edges initialized to the complete bipartite graph $G$ \;
   \While {$E \ne \emptyset$} {
     Pick the heaviest edge $e=(x,y)$ from E and add it to $M$  \;
     Remove $x$, $y$, and all edges containing $x$ or $y$ from $E$ \;
     \If { $|M| = k$ }  {
     	break \;
      }
    }
  \textbf{Final Output:} Return $M$.
\end{algorithm}
\vspace{3mm}

\begin{lemma}
\label{partial-totalorder-greedy-k}
  Suppose $G=(\mathcal{X}, \mathcal{Y}, E)$ is a complete bipartite graph on the set of nodes $\mathcal{X}, \mathcal{Y}$ with $|\mathcal{X}| = |\mathcal{Y}| = N$. Given $k = \gamma N$, the performance of the greedy k-matching returned by Algorithm~\ref{total-order-greedy} with respect to the optimal perfect matching OPT is at least $\gamma$, for $\gamma\leq\frac{1}{2}$.
\end{lemma}
\begin{proof}
  Let $M$ be the matching returned by Algorithm~\ref{twosided-greedy} for $k = N$. From Lemma~\ref{twosided-greedy-k}, $w(M) \ge \frac{1}{2}w(OPT)$. In the proof of Lemma~\ref{twosided-greedy-k}, each edge in $M$ is charged at most twice by edges of OPT, and there are $N$ charges in total. Transfer all the charges to the highest weight $\frac{N}{2}$ edges in $M$; this tells us that the highest weight $\frac{N}{2}$ edges of $M$ are at least $\frac{1}{2}w(OPT)$. Further transfer all the charges to the highest weight $\gamma N$ edges in $M$; this results in each such edge being charged to $1/\gamma$ times by edges of OPT. Therefore,  the highest weight $\gamma N$ edges of $M$ are at least $\frac{1}{\gamma}w(OPT)$ in total.

%
%
%
%
  Same as Algorithm~\ref{twosided-greedy}, Algorithm~\ref{total-order-greedy} also picks an undominated edge each round; the difference is the edges in the matching are picked in non-decreasing order. So Algorithm~\ref{total-order-greedy} returns a $k$-matching with the same weight as the highest $\gamma N$ edges in the perfect matching returned by Algorithm~\ref{twosided-greedy} on the same graph, which gives at least $\frac{1}{\gamma}w(OPT)$.
\end{proof}

\begin{lemma}
\label{partial-totalorder-greedy-alpha}
  Suppose $G=(\mathcal{X}, \mathcal{Y}, E)$ is a complete bipartite graph on the set of nodes $\mathcal{X}, \mathcal{Y}$ with $|\mathcal{X}| = |\mathcal{Y}| = N$. Given the order of the top $\alpha N^2$ edges in the graph, we are able to run greedy k-matching by Algorithm~\ref{total-order-greedy} for $k = (1 - \sqrt{1 - \alpha})N$.
\end{lemma}
\begin{proof}
  In the first step of Algorithm~\ref{total-order-greedy}, the heaviest edge is taken, and $2N-1$ edges are removed, so at most $2N-1$ edges are lost from the top $\alpha N^2$ edges. After the first $k$ steps of Algorithm~\ref{total-order-greedy}, the total number of removed edges is:
  \begin{align*}
    & 2N - 1 + 2(N-1) -1 + ... + 2(N-(k-1)) - 1 \\
    &= 2(N + N-1 + ... + N-(k-1)) - k\\
    &= 2Nk - k^2
  \end{align*}

  \noindent Given the order of top $\alpha N^2$ edges, we are able to run Algorithm~\ref{total-order-greedy} for at least $k$ steps until $2Nk - k^2 = \alpha N^2$. Solve the equation for $k$, $k = (1-\sqrt{1-\alpha})N$.
\end{proof}

The algorithm for bipartite matching with partial ordinal information is similar to that with partial two-sided ordinal information, except that we only need to consider the case that $k \leq \frac{1}{2}N$, i.e., $1-\sqrt{1-\alpha} \le \frac{1}{2}$, $\alpha \le \frac{3}{4}$. In two-sided model, we are given the top $\alpha N$ preferences for both sets of agents, and able to run greedy algorithm for $k = \alpha N$. While in total ordering model, we could only run greedy algorithm for  $k = (1-\sqrt{1-\alpha})N$ given the order of the top $\alpha N^2$ edges. Different from two-sided model, $\alpha$ does not equal to the number of agent pairs we are able to match by greedy algorithm in total ordering model.
%
%

\vspace{3mm}
\begin{algorithm}[H]
\caption{Algorithm for matching given partial total ordering.}
\label{partial-total-order-alg}
  \SetKwInOut{Input}{Input}
  \SetKwInOut{Output}{Output}
  \Input{$\mathcal{X}, \mathcal{Y}$, order of the top $\alpha N^2$ edges in the graph.}
  \Output{ Perfect Bipartite Matching M }
  Initialize $E$ to be complete bipartite graph on $\mathcal{X}, \mathcal{Y}$, and $M_1 = M_2 = \emptyset $ \;
  Let $M_0$ be the output returned by Algorithm~\ref{total-order-greedy} for $E$, $k = (1 - \sqrt{1 - \alpha}) N$. Let $\alpha_1 = 1 - \sqrt{1 - \alpha}$, then $k = \alpha_1 N$ \;
  Let $\mathcal{X}_T$ be the set of nodes in  $\mathcal{X}$ matched in $M_0$, $\mathcal{Y}_T$ be the set of nodes in $\mathcal{Y}$ matched in $M_0$, and $T$ be the complete bipartite graph on $\mathcal{X}_T, \mathcal{Y}_T$\ \;
  Let $\mathcal{X}_B$ be the set of nodes in $\mathcal{X}$ not matched in $M_0$, $\mathcal{Y}_B$ be the set of nodes in $\mathcal{Y}$ not matched in $M_0$, and $B$ is the complete bipartite graph on $\mathcal{X}_B, \mathcal{Y}_B$\;

  \textbf{First Algorithm};\\
  $M_1 = M_0 \cup$ (The perfect matching output by Algorithm~\ref{random} on $B$)\;

  \textbf{Second Algorithm};\\

  Choose $(1-2\alpha_1)N$ nodes both from $\mathcal{X}_B$ and $ \mathcal{Y}_B$ uniformly at random, get the perfect matching output by Algorithm~\ref{random}  on these nodes and add the results to $M_2$ \;
  Let $X_A$ be the set of nodes in $ \mathcal{X}_B$ and not in $M_2$,  $Y_A$ be the set of nodes in $ \mathcal{Y}_B$ and not in  $M_2$\;
  Let $E_{AT}$ be the edges of the complete bipartite graph $(X_A, \mathcal{Y}_T)$ and $E_{AT}'$ be the edges of the complete bipartite graph $(\mathcal{X}_T, Y_A)$ \;
  Run random bipartite matching on the set of edges in $E_{AT}$ and $E_{AT}'$ separately to obtain perfect bipartite matchings and add the edges returned by the algorithm to $M_2$\;

  \textbf{Final Output:} Return $M_1$ with probability $\frac{2}{2 + \sqrt{1 - \alpha}}$ and $M_2$ with probability $\frac{\sqrt{1 - \alpha}}{2 + \sqrt{1 - \alpha}}$.
\end{algorithm}
\vspace{3mm}

\begin{theorem}
\label{total-order}
  Algorithm \ref{partial-total-order-alg} returns a $\frac{2 + \sqrt{1 - \alpha}}{2 - \sqrt{1 - \alpha}}$-approximation to the maximum-weight matching in expectation for $\alpha\leq \frac{3}{4}$, as shown in Figure~\ref{fig:plot_figures}.
\end{theorem}

  \begin{proof}
    By Lemma \ref{partial-totalorder-greedy-alpha}, we are able to run Algorithm~\ref{total-order-greedy} for $k = (1 - \sqrt{1 - \alpha}) N$. We analyze the algorithm when $\alpha \le \frac{3}{4}$, $\alpha_1 = 1 - \sqrt{1 - \alpha} \le \frac{1}{2}$.\\

    $ |\mathcal{X}_T| = |\mathcal{Y}_T| = \alpha_1 N$, $ |\mathcal{X}_B| = |\mathcal{Y}_B| = (1-\alpha_1)N$.\\

    By Lemma~\ref{partial-totalorder-greedy-k}, $w(M_0) \ge \alpha_1 w(OPT)$. By Lemma~\ref{random-bound}, the perfect matching output by Algorithm~\ref{random}  on $B$ has expected weight $\frac{1}{(1-\alpha_1)N}w(B)$. Thus,
    \begin{equation*}
    E[w(M_1)] \ge \alpha_1 w(OPT) + \frac{1}{(1-\alpha_1)N} w(B)
    \end{equation*}

    Analysis of $E[w(M_2)]$ is very similar to the case when $\alpha_1\geq \frac{1}{2}$ for Algorithm~\ref{partial-twosided-alg}, except that now $B$ is larger than $T$, and so we form a random bipartite matching using all of the nodes in $T$ instead of just some of them. Formally, because $|X_A| = |Y_A| = \alpha_1 N$, and they are leftover nodes after $(1-2\alpha_1)N$ nodes are chosen uniformly at random from $B$, we know that
    \begin{equation*}
    E[w(E_{AT}) + w(E_{AT}^{'})] = \frac{\alpha_1}{1-\alpha_1} w(T, B).
    \end{equation*}

    Let $M_{AT}$ be the random bipartite matching formed between sets $A$ and $T$. By Lemma~\ref{random-bound},
    \begin{align*}
    E[w(M_{AT})] &= \frac{1}{\alpha_1 N}E[w(E_{AT})] + \frac{1}{\alpha_1 N}E[w(E_{AT}^{'})] \\
          &= \frac{1}{(1 - \alpha_1)N} w(T, B)
    \end{align*}

    By Lemma~\ref{upper-bound}, setting $M=OPT, T = B, n=(1 - \alpha_1)N$,
    \begin{equation*}
    (1 - \alpha_1)N w(OPT) \le (2 + \frac{1}{1 - \alpha_1}) w(B) + w(T, B).
    \end{equation*}

  Thus,
    \begin{align*}
     E[w(M_{AT})] &= \frac{1}{(1 - \alpha_1)N} w(T, B)\\
                  &\ge \frac{1}{(1 - \alpha_1)N}  ((1 - \alpha_1)N w(OPT) - \frac{3 - 2\alpha_1}{1-\alpha_1}w(B) )
   \end{align*}

    \begin{align*}
     E[w(M_2)] &= \frac{1 - 2\alpha_1}{1 - \alpha_1} \times \frac{1}{(1 - \alpha_1) N} w(B) +  E[w(M_{AT})]\\
     &\ge \frac{1 - 2\alpha_1}{(1 - \alpha_1)^2 N}w(B) + \frac{1}{(1 - \alpha_1) N}  ((1 - \alpha_1)N w(OPT) - \frac{3 - 2\alpha_1}{1-\alpha_1}w(B) )\\
     &= w(OPT) - \frac{2}{ (1-\alpha_1)^2 N}w(B)
    \end{align*}


    Return $M_1$ with probability $\frac{2}{3-\alpha_1} = \frac{2}{2 + \sqrt{1 - \alpha}}$, and $M_2$ with probability  $\frac{1 - \alpha_1}{3 - \alpha_1} = \frac{\sqrt{1 - \alpha}}{2 + \sqrt{1 - \alpha}}$,
    \begin{align*}
    \frac{2}{3-\alpha_1}E[w(M_1)] + \frac{1-\alpha_1}{3 - \alpha_1}E[w(M_2)] &\ge \frac{1 + \alpha_1}{3 - \alpha_1} w(OPT)\\
    &=\frac{2 - \sqrt{1 - \alpha}}{2 + \sqrt{1 - \alpha}} w(OPT)
    \end{align*}
  \end{proof}

\section{One-sided Preferences with Restricted Edge Weights}
\label{section-onesided-restricted}
In previous sections, we made the assumption that the agents lie in a metric space, and thus the edge weights, although unknown to us, must follow the triangle inequality. In this section we once again consider the most restrictive type of agent preferences --- that of one-sided preferences --- but now instead of assuming that agents lie in a metric space, we instead consider settings where edges weights cannot be infinitely different from each other. This applies to settings where the agents are at least somewhat indifferent and the items are somewhat similar; the least-preferred agent and the most-preferred items differ only by a constant factor to any agent. Indeed, when for example purchasing a house in a reasonable market (i.e., once houses that almost no one would buy have been removed from consideration), it is unlikely that any agent would like house $x$ so much more than house $y$ that they would be willing to pay hundreds of times more for $x$ than for $y$.

More formally, for each agent $i \in \mathcal{X}$, we are given a strict preference ordering $P_i$ over the agents in $\mathcal{Y}$. In this section we assume that the highest weight edge $e_{max}$ is at most $\beta$ times of the lowest weight edge $e_{min}$. We normalize the lowest weight edge $e_{min}$ in the graph to $w(e_{min}) = 1$; then for any edge $e \in E$, $w(e) \le \beta$. We use similar analysis as in Section~\ref{section-onesided}, except that instead of getting bounds by using the triangle inequality, the relationships among edge weights are bounded by our assumption of the highest and lowest weight edge ratio. As stated above, we no longer assume the agents lie in a metric space in this section.
\vspace{5mm}

\begin{theorem}
\label{onesided-rsd-bound2}
  Suppose $G=(\mathcal{X}, \mathcal{Y}, E)$ is a complete bipartite graph on the set of nodes $\mathcal{X}, \mathcal{Y}$ with $|\mathcal{X}| = |\mathcal{Y}| = N$. $w(e_{min}) = 1$, $\forall e \in E$, $w(e) \le \beta$. The expected weight of the perfect matching returned by Algorithm~\ref{onesided-rsd} is $w(M) \ge \frac{1}{\sqrt{\beta - \frac{3}{4}} + \frac{1}{2}} w(OPT)$.
\end{theorem}

\begin{figure}[H]
\begin{center}
\includegraphics[scale=0.5]{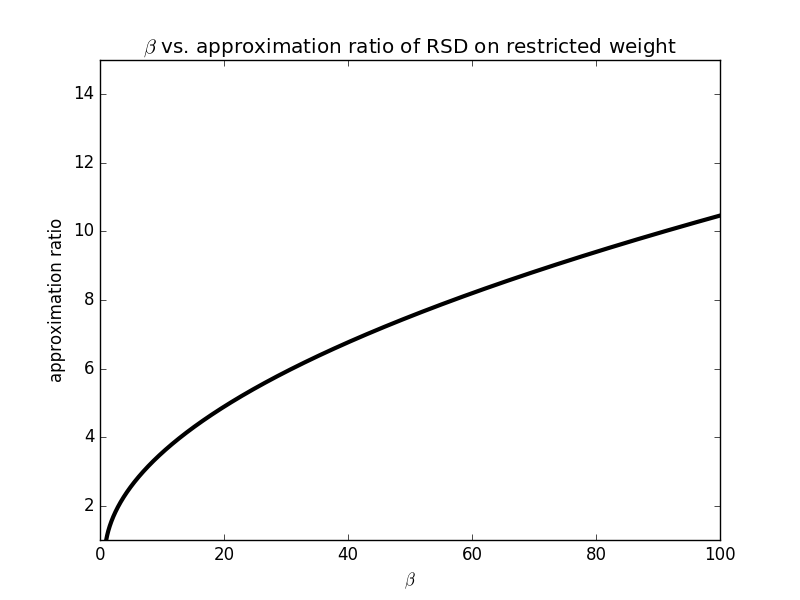}
\end{center}
\caption{$\beta$ vs. approximation ratio of RSD on restricted weight bipartite graph. For edges with a small difference in weight, we still obtain a reasonable approximation to the optimum matching.}
\label{fig:rsd-restricted}
\end{figure}

\begin{proof}
  We use the same notation as in Section~\ref{section-onesided}. Once again, our proof relies on the following claim, similar to Lemma \ref{onesided-structural}. Once the statement below is proven, the rest of the proof proceeds exactly as in Theorem \ref{onesided-rsd-bound}, simply replacing $\sqrt{2} + 1$ with $\sqrt{\beta - \frac{3}{4}} + \frac{1}{2}$.

  \begin{lemma}
    \label{onesided-structural-restricted}
    For any given subgraph $S = (\mathcal{X}', \mathcal{Y}', E')$, one of the following two cases must be true:\\
    \textbf{Case 1},  $w(OPT(S)) \le \frac{1}{|\mathcal{X}'|} \sum_{x \in \mathcal{X}'} w(OPT(R(S, x))) + \frac{\sqrt{\beta - \frac{3}{4}} + \frac{1}{2}}{|\mathcal{X}'|}\sum_{x \in \mathcal{X}'} w(\lambda(x))$\\
    \textbf{Case 2}, $w(OPT(S)) \le (\sqrt{\beta - \frac{3}{4}} + \frac{1}{2})w(Min(S))$\\
  \end{lemma}

%

\begin{proof}
  Again, we use the same notation as in Section~\ref{section-onesided}.

  We'll prove Lemma~\ref{onesided-structural-restricted} by showing that if \textbf{Case 2} is not true, then  \textbf{Case 1} must be true. Suppose  \textbf{Case 2} is not true, $w(OPT(S)) > (\sqrt{\beta - \frac{3}{4}} + \frac{1}{2})w(Min(S))$. \\

  Suppose that random serial dictatorship picks $x \in \mathcal{X}'$. Just as in the proof of Lemma \ref{onesided-structural}, we obtain that

  \begin{align}
  \frac{1}{|\mathcal{X}'|} \sum_{x \in \mathcal{X}'} w(OPT(R(x)))
  & \ge (1-\frac{1}{|\mathcal{X}'|})w(OPT(S)) - \frac{1}{|\mathcal{X}'|} \sum_{x \in \mathcal{X}'} ( w(\bar{P}(x)) - w(D(x))) \label{rsd-restricted-eq1}
  \end{align}

  We know that $\forall e \in E'$, $1 \le w(e) \le \beta$. So $w(D(x)) \ge 1$, $w(\bar{P}(x)) \le \beta$, and thus
  \begin{equation}
  \label{rsd-restricted-eq2}
  \frac{1}{|\mathcal{X}'|} \sum_{x \in \mathcal{X}'} ( w(\bar{P}(x)) - w(D(x))) \le \beta - 1
  \end{equation}

  $\forall x \in \mathcal{X}'$, $w(P(x)) \le w(\lambda(x))$, so it is obvious that $w(OPT(S)) \le \sum_{x \in \mathcal{X}'} w(\lambda(x))$.\\

  $Min(S)$ is a perfect matching, so $w(Min(S)) \ge |\mathcal{X}'|$. By our assumption,
  \begin{equation}
  \label{rsd-restricted-eq3}
  |\mathcal{X}'| \le w(Min(S)) < \frac{1}{\sqrt{\beta - \frac{3}{4}} + \frac{1}{2}} w(OPT(S))
  \end{equation}

  Combining Inequalities ~\ref{rsd-restricted-eq1}, \ref{rsd-restricted-eq2}, and \ref{rsd-restricted-eq3},

  \begin{align*}
    \frac{1}{|\mathcal{X}'|} \sum_{x \in \mathcal{X}'} w(OPT(R(x)))
    & \ge w(OPT(S)) - \frac{1}{|\mathcal{X}'|} w(OPT(S)) - \frac{1}{|\mathcal{X}'|}(\beta - 1)|\mathcal{X}'| \\
    & \ge w(OPT(S)) - \frac{1}{|\mathcal{X}'|} w(OPT(S)) - \frac{1}{|\mathcal{X}'|}(\beta - 1) \frac{1}{\sqrt{\beta - \frac{3}{4}} + \frac{1}{2}} w(OPT(S))\\
    &= w(OPT(S)) - \frac{1}{|\mathcal{X}'|}(1 + \frac{\beta - 1}{\sqrt{\beta - \frac{3}{4}} + \frac{1}{2}})w(OPT(S))\\
    &= w(OPT(S)) - \frac{\sqrt{\beta - \frac{3}{4}} + \frac{1}{2}}{|\mathcal{X}'|}w(OPT(S))\\
    &\ge w(OPT(S)) - \frac{\sqrt{\beta - \frac{3}{4}} + \frac{1}{2}}{|\mathcal{X}'|}\sum_{x \in \mathcal{X}'} w(\lambda(x))
  \end{align*}
\end{proof}

This completes the proof of the theorem.
\end{proof}

\section{Lower Bound Examples}
\label{section_lower_bound}
In this section, we provide some example to study the lower bound of algorithms on maximum weight bipartite graph perfect matching, given two-sided or one-sided ordinal information.

\subsection{Lower Bound of Two-sided Ordinal Information}
\label{lower_bound_twosided}
\textbf{Example} Consider a bipartite graph $G=(\mathcal{X}, \mathcal{Y}, E)$, $\mathcal{X} = \{a, b\}$,  $\mathcal{Y} = \{c, d\}$. Let $\epsilon$ be a very small positive number. Consider two sets of weight assignment that have the same two-sided ordinal preferences in metric space: $W1: w(a,c) = w(b, d) = 1 + \epsilon$, $w(b, c) = 3$, $w(a, d) = 1$. $W2: w(a,c) = w(b, d) = 1 - \epsilon$, $w(b, c) = 1$, $w(a, d) = \epsilon$. The maximum weight perfect matching for $W1$ is $M1 = \{(a,d), (b,c)\}$, while for $W2$ is $M2 = \{(a,c), (b,d)\}$. Applying any randomized algorithm choosing $M1$ with probability $p$ and $M2$ with probability $1-p$ to these two weight settings, the optimal algorithm is when $p = \frac{1}{2}$, gives a $1.33$-approximation.

\subsection{Lower Bound of One-sided Ordinal Information}
\label{lower_bound_onesided}
\textbf{Example}
For one-sided ordinal information, consider a graph  $G=(\mathcal{X}, \mathcal{Y}, E)$, $|\mathcal{X}|$ = $|\mathcal{Y}| = N$,  $\mathcal{X} = \{x_1, x_2, ... x_N\}$,  $\mathcal{Y} = \{y_1, y_2, ..., y_N\}$. Each agent in $\mathcal{X}$ have the same preferences over agents in $\mathcal{Y}$ as $y_1 > y_2 > ... > y_N$, because of this setting, no random algorithm could distinguish agents and get a better performance than random algorithm. Assign the weights of the graph as: for a certain number $\nu \in [0, 1]$, when $i <= \nu$, $w(x_i, y_j) = 3$ for $j <= i$, all other edges have weight $1$. The maximum matching is $\{ (a_1, b_1), (a_2, b_2), ..., (a_N, b_N)\}$, with a total weight $(2\nu + 1)N$. Random matching of this graph gets an expected weight of $(\nu(\frac{1}{N} + \nu)+1)N$, when $N$ is large, the weight approaches $(\nu^2+1)N$. When $\nu = \frac{\sqrt{5} - 1}{2}$, random algorithm gets a $1.62$-approximation, which is a lower bound of one-sided ordinal information setting.

\section{Conclusion}
In this paper we quantified the tradeoffs between the amount of ordinal information available, and the quality of solutions produced by our ordinal approximation algorithms, for metric maximum-weight bipartite matchings. For example, if we are able to collect preference data through surveys, but for each extra preference we must perform a certain extra amount of market research (i.e., increasing $\alpha$ comes at a cost), then our findings would quantify how big we should make $\alpha$ in order to form a good approximation to the best possible matching. All of this is without knowing the true numerical weights, only ordinal information.

One thing to note here is that asking people to list their preference orderings, even partial preference orderings for relatively small $\alpha$, may be prohibitive. Agents are usually willing to name their top 3-10 choices, but not more than that. Notice, however, that all our algorithms can be thought of differently. For example, RSD does not actually require the preference ordering as an input. It simply needs to ask each agent a single question: what is you favorite agent who has not been matched yet? Similarly, our other algorithms can be considered to ask agents a series of questions about their preferences, all of the same form. Such questions (determining their favorite from a set) are usually much easier for agents to answer than the question of specifying a preference ordering.

One clear research direction is to relax the assumption that we can only obtain ordinal information. What if we could also obtain some numerical information, but at further cost? What is the tradeoff between quality of solution formed and the amount of numerical information we obtain? What if we could ask the agents more complex questions than ``Who is your favorite unmatched agent?", but were limited in the number of times we could ask such questions? We leave these important directions for future work.

\bibliographystyle{plain}
\bibliography{ref}

\end{document}